\documentclass[9pt,numbers,times]{sigplanconf}
\usepackage{etex}
\usepackage{amsmath, amssymb}
\usepackage{times}
\usepackage{balance}

\usepackage{nicefrac}
\usepackage{proof}
\usepackage{amsthm} \usepackage{stmaryrd}
\usepackage{golang}
\usepackage[final]{listings}
\usepackage{mathpartir}
\usepackage[showonlyrefs]{mathtools}
\usepackage{tipa}
\usepackage[final]{fixme}
\fxsetup{theme=color}
\definecolor{fxtarget}{rgb}{0.80000,0.0000,0.0000}
\usepackage[all,cmtip]{xy}
\usepackage{color}
\usepackage[dvipsnames]{xcolor}
\usepackage{tikz}\usetikzlibrary{arrows,shapes,calc,fit,patterns,backgrounds,trees,arrows,positioning,shapes.multipart}
\usepackage{url}
\usepackage{xspace}
\usepackage{thm-restate}
\usepackage{bm}
\usepackage{booktabs}
\usepackage{etoolbox}
\usepackage{hyperref}
\newtheorem{lemma}{Lemma}[section]
\newtheorem{theorem}{Theorem}[section]

\newtheorem{proposition}{Proposition}[section]
\newtheorem{remark}{Remark}[section]

\newtheorem{notation}{Notation}[section]

\theoremstyle{definition}

\newtheorem{definition}{Definition}[section]

\renewcommand{\implies}{\quad \Longrightarrow \quad}

\newcommand{\naturals}{\mathbb{N}}

\newcommand{\num}{\mathsf{n}}
\newcommand{\fv}{\mathsf{fv}}
\newcommand{\fn}{\mathsf{fn}}
\newcommand{\allnames}{\mathsf{n}}

\DeclareMathAlphabet{\mathbfsf}{\encodingdefault}{\sfdefault}{bx}{sl}

\newcommand{\freenames}[1]{\fn(#1)}

\newcommand{\subs}[2]{\left\{\nicefrac{#1}{#2}\right\}}
\newcommand{\qst}{\,:\,}
\newcommand{\defi}{\overset{{\text{def}}}{=}}
\newcommand{\domain}[1]{\mathsf{dom}(#1)}

\newcommand{\m}[1]{\mathsf{#1}}
\newcommand{\ov}[1]{\overline{#1}}

\newcommand{\nub}{\mathbf{\nu}}
\newcommand{\tra}[1]{\xrightarrow{#1}}

\newcommand{\with}{\mathbin{\binampersand}}

\newcommand{\G}{\Gamma}
\newcommand{\D}{\Delta}

\newcommand{\vval}{v}
\newcommand{\expr}{e}

\newcommand{\psend}[2]{#1 ! \langle #2 \rangle }
\newcommand{\ENCan}[1]{\langle #1 \rangle}
\newcommand{\precv}[2]{#1 ? (#2) }

\newcommand{\zero}{\mathbf{0}}

\newcommand{\prog}{\mathbfsf{P}}
\newcommand{\pin}[2]{#1 \, \m{in} \, #2}
\newcommand{\pite}[3]{\m{if}\, #1 \, \m{then} \, #2 \, \m{else}\, #3}
\newcommand{\mpite}[4]{\m{if}^{#4}\, #1 \, \m{then} \, #2 \, \m{else}\, #3}

\newcommand{\sel}[2]{\m{select}\{#1\}_{#2} }
\newcommand{\newch}[1]{\m{newchan}(#1)}
\newcommand{\true}{\m{true}}
\newcommand{\false}{\m{false}}

\newcommand{\close}[1]{\m{close}\,#1}
\newcommand{\newdef}[3]{#1(#2) = #3}
\newcommand{\defcall}[2]{#1\langle #2 \rangle}

\newcommand{\buff}[3]{{#2}\ENCan{{#3}}{::}#1}
\newcommand{\closedbuff}[3]{{#2}^\star\ENCan{#3}{::}{#1}}

\newcommand{\buflen}[1]{|#1|}

\newcommand{\parr}{\mid}

\newcommand{\tysend}[1]{\overline{#1}}
\newcommand{\tyrecv}[1]{#1} 
\newcommand{\tsend}[2]{\tysend{#1}} \newcommand{\trecv}[2]{\tyrecv{#1}}

\newcommand{\tch}[3]{#1 \oplus \{#2\}_{#3}}
\newcommand{\tbr}[3]{#1 \with \{#2\}_{#3}}

\newcommand{\tdefcall}[2]{\typevarfont{#1}\langle #2 \rangle}
\newcommand{\End}{\mathsf{end}}

\newcommand{\ts}{\blacktriangleright}

\newcommand{\fgo}{{\sf MiGo}\xspace}

\newcommand{\markred}[1]{\tra{#1}}
\newcommand{\emp}{\epsilon}

\newcommand{\llab}{\m{L}}
\newcommand{\rlab}{\m{R}}
\newcommand{\iflab}[2]{#1 \!\cdot\! #2}

\newcommand{\trace}{\mathcal{T}}
\newcommand{\traceset}{\mathbb{T}}

\newcommand{\barb}[1]{\downarrow_{#1}}
\newcommand{\wbarb}[1]{\Downarrow_{#1}}
\newcommand{\wbarbsym}[2]{\Downarrow^{#2}_{#1}}
\newcommand{\wbarbt}[1]{\Downarrow^{}_{#1}}

\newcommand{\clbarb}[1]{\labclose{#1}}
\newcommand{\cclbarb}[1]{\labclsnd{#1}}

\newcommand{\ac}{\m{AC}}
\newcommand{\ic}{\m{Inf}}
\newcommand{\convm}{\m{May}{\mbox{\scriptsize $\Downarrow$}}}
\newcommand{\infcond}{\m{InfCond}}

\newcommand{\EqTypes}{\mathbfsf{T}}
\newcommand{\typrefix}{\kappa}
\newcommand{\labtysnd}[1]{\overline{#1}}
\newcommand{\labtyrcv}[1]{{#1}}

\newcommand{\tychana}{a}
\newcommand{\tychanb}{b}
\newcommand{\tychanc}{c}
\newcommand{\tychand}{d}
\newcommand{\tychavar}{u}
\newcommand{\tyvar}{x}

\newcommand{\typevarfont}[1]{\mathbfsf{#1}}
\newcommand{\typevara}{\typevarfont{t}}
\newcommand{\typevarb}{\typevarfont{s}}

\newcommand{\labcom}{\alpha} 
\newcommand{\tyclosedbuffer}[1]{#1^\star}
\newcommand{\tyopenbuffer}[1]{\lfloor #1 \rfloor}
\newcommand{\tynew}[1]{(\m{new} \,#1)}
\newcommand{\labsync}[1]{[#1]}
\newcommand{\labclose}[1]{\End[#1]}
\newcommand{\labclosedual}[1]{\overline{\End}[#1]}
\newcommand{\labclsnd}[1]{\tyclosedbuffer{#1}}
\newcommand{\labspawn}[1]{\mathit{sp}}
\newcommand{\semty}[1]{\xrightarrow{#1}}

\newcommand{\closedchans}{{N}} \newcommand{\typesem}[2]{#2}
\newcommand{\typesemF}[1]{\typesem{\closedchans}{#1}}
\newcommand{\typesemFone}[1]{\typesem{\closedchans_1}{#1}}
\newcommand{\typesemFtwo}[1]{\typesem{\closedchans_2}{#1}}
\newcommand{\typesemFdash}[1]{\typesem{\closedchans'}{#1}}

\DeclareMathOperator{\tytyenvsep}{;}\newcommand{\tytyProc}{G}
\newcommand{\tytyRes}{\tilde{y}}
\newcommand{\tytyUnr}{\tilde{z}}
\newcommand{\tytyenvdef}{\Delta}

\newcommand{\tyfenced}[1]{\mathsf{Fenced}{(#1)}}
\newcommand{\tytyJudgeE}[1]{\tyfenced{#1}}
\newcommand{\tytyJudgeX}[5]{#5 \tytyenvsep #1 \tytyenvsep #2 \vdash_{#3} #4}
\newcommand{\tytyJudge}[3]{\tytyJudgeX{#1}{#2}{\typevara}{#3}{\tytyProc}}
\newcommand{\tytyJudgeDef}[1]{\tytyenvdef \vdash_{\typevara} #1}
\newcommand{\tytyRule}[3]{
  \frulename{#1}\inferrule[]
  {#2}
  {#3}
}
\DeclareMathOperator{\tytyprec}{\prec}
\DeclareMathOperator{\tytysucc}{\succ}

\newcommand{\typesymsem}[2]{#1 \lhd #2}
\newcommand{\typesymsemF}[1]{\symsemCur \lhd #1}
\newcommand{\symsemidx}{k}

\newcommand{\symsemCur}{N}
\newcommand{\symsemI}[2]{\xrightarrow{#1}_{#2}}
\newcommand{\symsem}[1]{\xrightarrow{#1}_{\symsemidx}}
\newcommand{\symsemTr}{\symsem{}^{\ast}}
\newcommand{\symsemTrI}[1]{\symsemI{}{#1}^{\ast}}
\newcommand{\typesymsemruleCI}[8]{
\symrulename{#1}\inferrule
  {#2}
  {\typesymsem{#3}{#4} \symsemI{#5}{#8} \typesymsem{#6}{#7}}
}
\newcommand{\typesymsemruleC}[7]{
  \typesymsemruleCI{#1}{#2}{#3}{#4}{#5}{#6}{#7}{\symsemidx}
}
\newcommand{\typesymsemrule}[5]{
  \typesymsemruleC{#1}{#2}{\symsemCur}{#3}{#4}{\symsemCur}{#5}}

\newcommand{\tyunfenv}{G}

\newcommand{\tysizeG}[2]{\lvert #1 \rvert^{#2}}
\newcommand{\tysizeD}[1]{\tysizeG{#1}{\tyunfenv}}
\newcommand{\tysize}[1]{\lvert #1 \rvert}

\newcommand{\tyocc}[2]{\lvert #1 \rvert_{#2}}
\newcommand{\tyocct}[1]{\lvert #1 \rvert_{\typevara}}
\newcommand{\tyunfX}[3]{\mathit{U}^{#1}_{#2}(#3)}
\newcommand{\tyunf}[1]{\tyunfX{\tyunfenv}{\tilde{\tychana}}{#1}}

\newcommand{\tycontxt}{\mathbb{C}}

\newcommand{\conv}{\downarrow}
\newcommand{\MAPAST}[1]{{#1}^\ast}

\newcommand{\resXdir}[2]{r{#1}\ensuremath{#2}}
\newcommand{\resoneless}{\resXdir{1}{<}}
\newcommand{\restwoless}{\resXdir{2}{<}}
\newcommand{\restwogeq}{\resXdir{2}{\geq}}
\newcommand{\resonegeq}{\resXdir{1}{\geq}}

\newcommand{\astyOpenbuf}[3]{\lfloor #1 \rfloor_{#2}^{#3}}
\newcommand{\astyBufRcv}[1]{#1^{\bullet}} \newcommand{\astyBufSend}[1]{^{\bullet} #1} \newcommand{\astynew}[2]{(\m{new}^{#2} \,#1)}

\newcommand{\trulename}[1]{\text{\scriptsize$\langle$\sc #1$\rangle$}\xspace}
\newcommand{\rulename}[1]{\text{\scriptsize[\sc #1]}\xspace}
\newcommand{\ltsrulename}[1]{\text{\scriptsize$|$\sc #1$|$}\xspace}
\newcommand{\frulename}[1]{\text{\scriptsize$\lfloor$\sc #1$\rfloor$}\xspace}
\newcommand{\symrulename}{\ltsrulename}

\newcommand{\primeChan}{a}
\newcommand{\primeN}[1]{\primeChan_{#1}}

\newcommand{\primeGN}{\typevarfont{g}}
\newcommand{\primeFN}{\typevarfont{f}}
\newcommand{\primeRN}{\typevarfont{r}}
\newcommand{\primeG}[1]{\primeGN \langle {#1} \rangle}
\newcommand{\primeF}[2]{\primeFN \langle {#1}, {#2} \rangle}
\newcommand{\primeR}[1]{\primeRN \langle {#1} \rangle}
\newcommand{\primecolorA}{blue}
\newcommand{\primecolorB}{red}
\newcommand{\primecolorC}{PineGreen}

\newcommand{\fibChan}{a}
\newcommand{\fibChanB}{b}

\newcommand{\fibGN}{\typevarfont{fib}}
\newcommand{\fibG}[1]{\fibGN \langle {#1} \rangle}
\newcommand{\fibocolorA}{blue}
\newcommand{\fibocolorB}{red}
\newcommand{\fibocolorC}{red}

\newcommand{\fibinftrees}[1]{
  child{
    node (#1dots1) {}   }
  child{
    node (#1dots2)  {}   }
  child{
    node (#1dots3) {}   }
}

\newcommand{\nofencechanColor}{blue}
\newcommand{\nofencechan}{{\color{\nofencechanColor}{a}}}

\newcommand{\infertool}[0]{\texttt{GoInfer}\xspace}\newcommand{\checktool}[0]{\texttt{Gong}\xspace}
 \usepackage[dvipsnames]{xcolor}

\newcommand{\lstCodeSize}{\scriptsize}
\newcommand{\lstPrimitiveStyle}{\color{Blue}\bfseries}

\newcommand{\lstNumberStyle}{\tiny\sffamily\color{Gray}}

\lstdefinestyle{nonumber}{  numbers=none,
  xleftmargin=1em,
  framexleftmargin=1em,
}
 
\newtoggle{cameraready}

\togglefalse{cameraready}

\usepackage[firstpage]{draftwatermark}
\SetWatermarkText{\hspace*{8in}\raisebox{3.8in}{\includegraphics[scale=0.1]{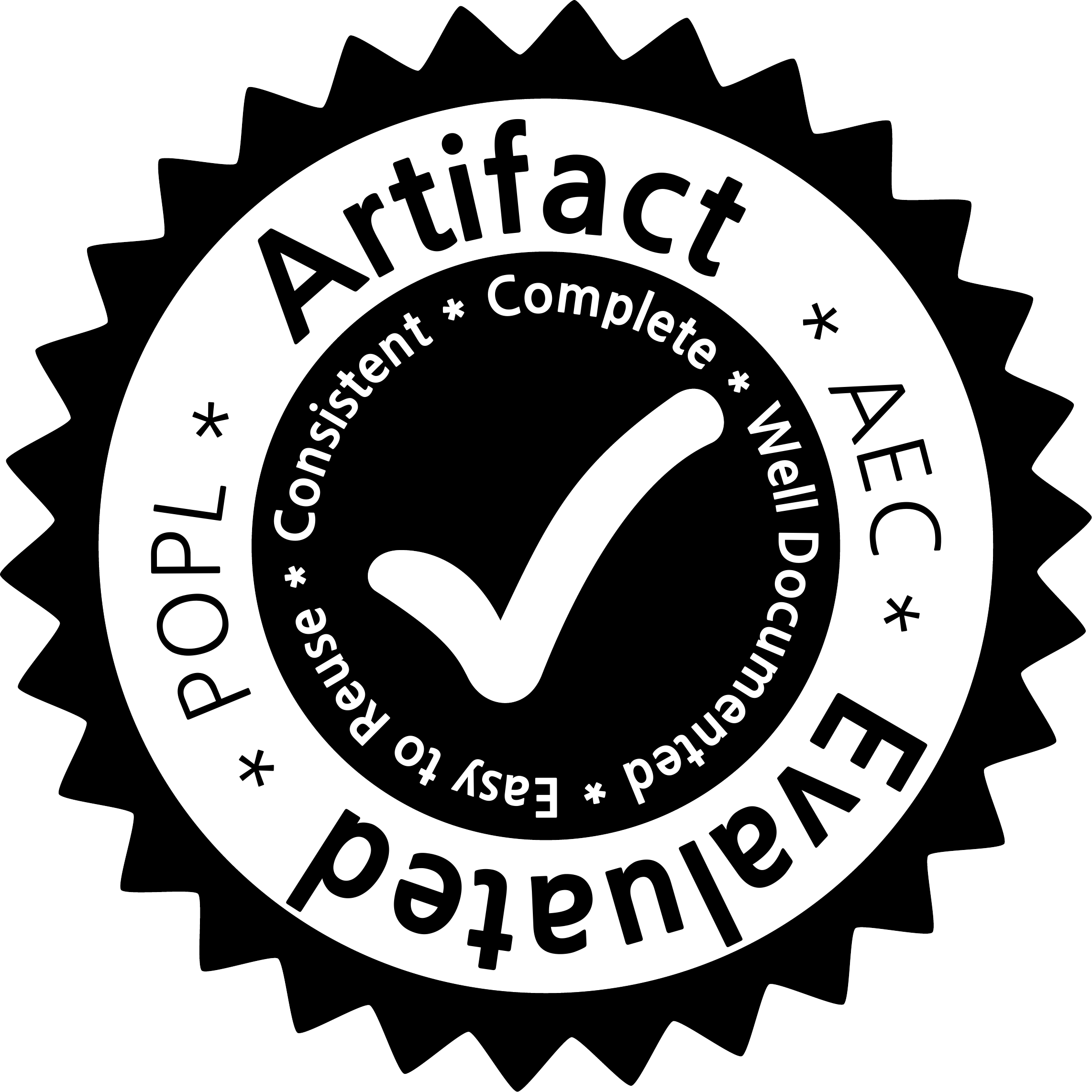}}}
\SetWatermarkAngle{0}
\begin{document}

\toappear{}

\setlength{\pdfpageheight}{\paperheight}
\setlength{\pdfpagewidth}{\paperwidth}

\title{Fencing off Go:\\
{\Large Liveness and Safety for Channel-Based Programming}}

\authorinfo{Julien Lange \and Nicholas Ng \and Bernardo Toninho \and Nobuko Yoshida}
{Imperial College London, UK}
           {\{j.lange, nickng, b.toninho, n.yoshida\}@imperial.ac.uk}

\maketitle

\begin{abstract}
  Go is a production-level statically typed programming language whose
  design features explicit message-passing 
  primitives and lightweight threads, enabling (and encouraging)
  programmers to develop concurrent systems where components interact
  through communication more so than by lock-based shared
  memory concurrency. 
  Go can only detect global deadlocks at runtime, 
  but provides no compile-time protection against 
  all too common communication mismatches or partial deadlocks.

  This work develops a static verification framework for bounded
  liveness and safety in Go programs, able to detect communication
  errors and partial deadlocks in a general class of realistic
  concurrent programs, including those with dynamic channel creation
  and infinite recursion.  Our approach infers from a Go program a
  faithful representation of its communication patterns as a
  behavioural type. By checking a syntactic restriction on channel
  usage, dubbed \emph{fencing}, we ensure that programs are made up of
  finitely many different communication patterns that may be repeated
  infinitely many times. This restriction allows us to implement
  bounded verification procedures (akin to bounded model checking) to
  check for liveness and safety in types which in turn approximates
  liveness and safety in Go programs.
We have implemented a type inference and liveness and safety checks
in a tool-chain and tested it against
publicly available Go programs.

\end{abstract}

\category{D.1.3}{Programming Techniques}{Concurrent Programming}
\category{D.2.4}{Software Engineering}{Software/Program Verification}
\category{D.3.1}{Programming Languages}{Formal Definitions and Theory}
\category{F.3.2}{Semantics of Programming Languages}{Program analysis}

\keywords 
Channel-based programming, Message-passing programming,
Process calculus, Types, Safety and Liveness, 
Compile-time (static) deadlock detection

\section{Introduction}
\begin{tabular}{lr}
\hspace{1.3cm} & {\em Do not communicate by sharing memory;}\\ 
\hspace{1.3cm}& {\em instead, share memory by communicating}\\
  & \hspace{2.2cm} {\sl Go language proverb~\cite{web:effective-go,web:share-by-comm}}\\[1mm]
\end{tabular}

\noindent 
Go is a statically typed programming language designed with explicit
concurrency primitives at the forefront, namely {\em channels} and
{\em goroutines} (i.e. lightweight threads), drawing heavily from
process calculi such as Communicating Sequential Processes (CSP)
\cite{book:csp}.  
Concurrent programming in Go is mostly guided towards
channel-based communication as a way to exchange 
data between goroutines, rather than more classical
concurrency control mechanisms (e.g.~locks).  
Channel-based concurrency in Go is lightweight, offering
logically structured flows of messages in large
systems programming \cite{flywheel:nsdi15,web:flywheel}, instead of
``messy chains of dozens of asynchronous callbacks spread over tens of
source files'' \cite{web:epaxos,epaxos:sosp13}.

On the other hand, Go inherits most problems commonly found in
concurrent message-passing programming such as communication
mismatches and deadlocks, offering very little in terms of
compile-time assurances of correct structuring of communication.
While the Go runtime includes a (sound) global deadlock detector, it is
ultimately inadequate for complex, large scale applications that may
easily be undermined by trivial mistakes or benign changes to the
program structure~\cite{NY2016,gh:golang:dld-failed}, nor can it
detect deadlocks involving only a strict subset of a program's
goroutines (partial deadlocks).

\paragraph{Liveness and Safety of Communication} While Go's type system does ensure
that channels are used to communicate values of the appropriate type,
it makes no static guarantees about the liveness or channel safety
(i.e.\ channels may be closed in Go, and sending on a closed channel
raises a runtime error) of communication in well-typed code.
Our work provides a framework for the \emph{static} verification of
liveness and absence of communication errors (i.e.\ safety) of Go
programs by extracting {\em concurrent behavioural types}
from Go source code (related type systems have been pioneered by~\cite{THK}
and~\cite{DBLP:journals/tcs/IgarashiK04}, among others).
Our types can be seen as an abstract representation of communication
behaviours in the given program. 
We then perform an analysis on behavioural types, which checks for a
bounded form of liveness and communication safety on types and study
the conditions under which our verification entails liveness and
communication safety of programs.

 \subsection{Overview} \label{sec:overview}
We present a general overview of the steps
needed to perform our analysis of concurrent Go programs. 

\begin{figure}
\begin{lstlisting}[
  breaklines=false,
  language=Go,style=golang,caption=Concurrent Prime Sieve.,label=lst:infinite-prime-sieve]
package main
func Generate(ch chan<- int) { /*<\label{lst:infinite-prime-sieve-gen:begin}>*/
  for i := 2; ; i++ { ch <- i } // Send sequence 2,3.../*<\label{lst:infinite-prime-sieve-gen:send}>*/
} /*<\label{lst:infinite-prime-sieve-gen:end}>*/
func Filter(in <-chan int, out chan<- int, prime int){/*<\label{lst:infinite-prime-sieve-fil:begin}>*/
  for { i := <-in // Receive value from 'in'.
    if i%prime != 0 { out <- i } // Fwd 'i' if factor./*<\label{lst:infinite-prime-sieve-filter}>*/
  }
} /*<\label{lst:infinite-prime-sieve-fil:end}>*/
func main() { /*<\label{lst:infinite-prime-sieve-main:begin}>*/
  ch := make(chan int) // Create new channel.
  go Generate(ch)      // Spawn generator. /*<\label{lst:infinite-prime-sieve-main:gen}>*/
  for i := 0; ; i++ {
    prime := <-ch /*<\label{lst:infinite-prime-sieve-main:prime}>*/
    ch1 := make(chan int)
    go Filter(ch, ch1, prime) // Chain filter.
    ch = ch1
  }
} /*<\label{lst:infinite-prime-sieve-main:end}>*/
\end{lstlisting}
\end{figure}

\paragraph{Prime Sieve in Go}  To illustrate the challenges in 
analysing Go programs, we begin by
considering a rather concise implementation of a concurrent prime
sieve (Listing~\ref{lst:infinite-prime-sieve}, adapting the example
from~\cite{web:concur-patterns} to output infinite primes). 
This seemingly simple Go program  
includes intricate communication 
patterns and concurrent behaviours that
are hard to reason about in general due to the combination
of (1) unbounded iterative behaviour, 
(2) dynamic channel creation
and (3) spawning of concurrent threads. 

The program is
made up of three functions: \verb=Generate=, that given a channel
\verb=ch= continuously sends along the channel an increasing sequence
of integers starting from $2$ (encoded as a \verb=for= loop without an
exit condition); \verb=Filter=, that given a channel for inputs
\verb=in=, one for outputs \verb=out=, and a \verb=prime=,
continuously forwards a number from \verb=in= to \verb=out= unless it
is divisible by \verb=prime=; and \verb=main=, which assembles the
sieve by creating a new synchronous channel \verb=ch=, spawning the
\verb=Generator= thread (i.e.\ \verb=go f(x)= spawns a parallel
instance of \verb=f=) with the channel \verb=ch=, and then iteratively
setting up a chain of \verb=Filter= threads, where the first filter
is connected to the generator and the next \verb=Filter=, and so on. We note
that with each iteration, a new synchronous channel \verb=ch1= is
created that is then used to link each filter instance.
The program spawns an infinite parallel
composition of \verb=Filter= threads, each pair connected by a
dynamically created channel; and the execution of
\verb=Generator= and the \verb=Filter= processes in the sieve is
non-terminating.

\paragraph{Types of Prime Sieve}  Our framework infers from the prime sieve
program the type $\typevara_0$ given by:
\[
\begin{array}{rcl}
\primeGN(x) & \triangleq & \tysend{x} ;\primeG{x}
\\
\primeFN(x,y) & \triangleq & \tyrecv{x};
(
\tysend{y}; \primeF{x}{y}
\ \oplus \
\primeF{x}{y}
)
\\
\primeRN(x) & \triangleq & \tyrecv{x};
\tynew{b}
(\primeF{x}{b} \parr \primeR{b})
\\
\typevara_0()  & \triangleq & \tynew{a} ( \primeG{a} \parr \primeR{a} )
\end{array}
\]
The type is described 
as a system of mutually recursive
equations, where the distinguished name $\typevara_0$ types the
program entry point. 
This language is equivalent to a subset of CCS \cite{MilnerR:calcs} with 
recursion, parallel and name creation, for which 
deciding liveness or safety is in general
\emph{undecidable}
\cite{DBLP:conf/icalp/BusiGZ03,DBLP:conf/icalp/BusiGZ04}. 
The type $\typevara_0$ specifies that
\verb=main= consists of the creation of a new channel $a$ and the two
parallel behaviours $\primeG{a}$ and $\primeR{a}$. The type given by equation
$\primeGN(x)$, which types the generator,  
identifies the infinite output behaviour along the given channel
$x$. The type equation $\primeFN(x,y)$ specifies the filter behaviour:
input along $x$ and then either output on $y$ and recurse or just
recurse. Finally, we decompose the topology set-up as
$\primeRN(x)$ which inputs along $x$ and then, through the creation of
a new channel $b$, spawns a filter on $x$ and $b$ and
recurses on $b$, creating an infinite parallel
composition of filters throughout the execution of the program.

Our type-level analysis relies on the fact that types are able to
accurately model a program's communication behaviour. The
analysis proceeds in two steps: 
a simple syntactic check on channel usage in types, 
dubbed \emph{fencing},
and a \emph{symbolic} finite-state execution
of \emph{fenced} types.\fxnote{Changed here.} 

\paragraph{Fenced Types}  Intuitively, if types are fenced, then  
during their execution there can only be a finite set of channels, or a
\emph{fence}, shared by finitely many types (threads). Moreover, fencing
ensures that recursive calls under parallel composition eventually
involve only local names, shared between finitely many threads.
This guarantees that a program consists of finitely many different
communication patterns (that may themselves be repeated infinitely
many times).

For instance, the recursive call to $\primeR{b}$ in the equation 
$\primeRN(x)$ is \emph{fenced}, since the recursive instances 
of $\primeRN$ will not know the channel parameter $x$
hence they cannot spawn threads which share $x$. 
This restriction
enforces a ``finite memory'' property wrt.\ channel names,
insofar as given enough recursive unfoldings all the parameters of the
recursive call will be names local to the recursive definition.

\paragraph{Symbolic Semantics and Bounded Verification} 
Fenced types can be \emph{symbolically} executed as CCS processes in a
representative \emph{finite-state} labelled transition system that
consists of a bounded version of the generally unrestricted
type semantics. 
This enables us to produce decision procedures for \emph{bounded}
liveness and safety of types, which deems the type
$\typevara_0$ as bounded live and safe. 
In the finite control case, 
the bounded correctness properties and their unbounded ones  
coincide. The approach is similar to bounded model checking, using a bound on
the number of channels in a type to limit its execution.

\paragraph{From Type Liveness to Program Liveness} The final step is to
formally relate liveness and safety of types to their analogues in Go
programs. For the case of safety, {safety of types
implies safety for programs.} 
{For the case of liveness,} 
programs typically rely on data values to guide their
control flow (e.g. in conditional branches) which are abstracted away
at the type level.  For instance, the \verb=Filter= function only
outputs if the received value \verb=i= is not divisible by the given
\verb=prime=, but the type for the corresponding process is given by
$\primeFN(x,y)$ which just indicates an internal choice between two
behaviours. The justification for the liveness of
$\typevara_0$ is that the internal choice always has the
\emph{potential} to enable the output on $y$ (assuming a fairness
condition on scheduling).  However, {it is not necessarily the case
that a}  conditional branch is ``equally likely'' to proceed to
the \verb=then= branch or the \verb=else= branch. 
{In \S~\ref{sec:properties}, we define 
the three classes of programs in which 
liveness of programs and their types coincide.}

\subsection{Contributions} 
We list the main contributions of our work:
To define and then show the bounded liveness and safety properties entailed by our
analysis, we formalise the message-passing concurrent fragment
of the Go language 
as a process calculus
(dubbed \fgo). 
The \fgo calculus mirrors very closely the
intended semantics of the channel-based Go constructs and allows us
to express highly dynamical and potentially complex behaviours  
({\bf \S~\ref{sect:lang}});
We introduce a typing system for \fgo 
which abstracts behaviours of \fgo as a subset of CCS process
behaviours ({\bf\S~\ref{sect:core_typ}});
We define a verification framework for our type language based on the
notion of \emph{fences} and \emph{symbolic execution}, showing that
for fenced types our symbolic semantics is finite control, entailing
the decidability of \emph{bounded} liveness and channel safety (i.e. the
notions of liveness and channel safety wrt.\ the symbolic
semantics -- {\bf \S~\ref{sec:types}})\fxnote{edited here};  We characterise the
\fgo programs whose liveness and safety are derived from our analysis on types ({\bf
  \S~\ref{sec:properties}}); We show that our results are
systematically extended to asynchronous communication semantics ({\bf
  \S~\ref{sec:asynchrony}});
we describe the implementation of our analysis in a tool
 that we use to evaluate our approach against
open-source, publicly available Go programs ({\bf \S~\ref{sec:impl}}).

\noindent
\iftoggle{cameraready}{
  The full proofs, omitted examples and definitions are available in
  appendix~\cite{APPENDIX}, while our implementation and benchmark
  examples are available online~\cite{TOOL}.
}{  
Our implementation and benchmark examples are available
 online~\cite{TOOL}.
}

 \section{\fgo: A Core Language for Go}\label{sect:lang}
This section introduces a core calculus that models
the message passing concurrency features of the Go programming
language, dubbed \fgo (mini-go). 
Beyond sending and receiving data values along channels,
the Go language supports three key concurrency features: 

{\bf FIFO Queues.}  Message-passing in Go is achieved via an abstract
notion of a lossless, order-preserving communication channel,
implemented as a (bounded) FIFO-queue.  When the bound on the queue
size is $0$, communication is \emph{fully synchronous}, whereas with
strictly positive bounds the communication is \emph{asynchronous}
(i.e.~sending is non-blocking if a queue is not full and, dually,
receiving is non-blocking if a queue is not empty). The bound is
defined upon channel creation and cannot be changed dynamically.

{\bf Goroutines.} 
Go supports lightweight threads, dubbed \emph{goroutines}, which denote the spawning of a thread to execute a
function concurrently with the main control flow of a program. 
This feature can be modelled by a combination of parallel composition
and process definitions. 

{\bf Select.} 
The $\m{select}$ construct in Go 
encodes a form of guarded choice, where each branch is guarded
by an input or an output on some channel. When multiple branches can
be chosen simultaneously, one is chosen at random (through
pseudo-random number generation). It is also possible to encode a
``timeout'' branch in such a choice construct.

Our fencing-based analysis is oblivious to buffer sizes, hence
we first focus on fully synchronous communication for ease of
presentation, addressing bounded asynchrony in
\S~\ref{sec:asynchrony}.

\subsection{Syntax of \fgo}\label{sect:core_syn}
The syntax of the calculus is given in Figure~\ref{fig:syntax}, where $P,Q$ range over
{\em processes}, $\pi$ over communication {\em prefixes}, $e,e'$ over
{\em expressions} and $x,y$ over {\em variables}.
We write $\tilde{\vval}$ and $\tilde{x}$ for a
list of expressions and variables, respectively
(we use $\cdot$ as a concatenation operator).
Programs are ranged over by $\prog$, consisting of a collection of
mutually recursive process definitions (ranged over by $D$),
parameterised by a list of (expression and channel) variables.
We omit a
detailed enumeration of types such as booleans, floating-point
numbers, etc., which are typed with payload types $\sigma$.
\begin{figure}
\[
\begin{array}{lcl}
P,Q & \coloneqq & \pi;P\\
    &  |  & \close{u};P\\
    &  |  & \sel{\pi_i;P_i}{i\in I}\\
    &  |  & \pite{\expr}{P}{Q}\\
    &  |  & \newch{y{:}\sigma};P \\
    &  |  & P \parr Q \ \mid \ \zero\\
    &  |  & \defcall{X}{\tilde \expr,\tilde{u}}\\
      &  |  & (\nu c) P \\
    &  |  & \buff{\tilde{\vval}}{c}{\sigma}  \mid  
\closedbuff{\tilde{\vval}}{c}{\sigma}
\end{array}
\ 
\begin{array}{lcl}
u & \coloneqq &  a \quad | \quad x
  \\[1mm]
\pi & \coloneqq & \psend{u}{\expr}  \mid  \precv{u}{y}
    \mid \tau\\[1mm]
  \vval & \coloneqq  & \num \mid \true \mid \false \mid x \\[1mm]
\expr & \coloneqq & \vval \mid \mathsf{not}(e) \mid \mathsf{succ}(e) \\[1mm]
    D   & \coloneqq & \newdef{X}{\tilde{x}}{P}\\[1mm]
\prog & \coloneqq & \pin{\{ D_i \}_{i\in I}}{P}\\[1mm]
\sigma & \coloneqq & \m{bool} \mid \m{int} \mid \dots\\[1em]
   \end{array}
\]
\caption{Syntax of \fgo.}\label{fig:syntax}
\end{figure}
$\fn(P)$ and $\fv(P)$ denote the sets of free names and variables.
Process variables $X$ are bound by {\em definitions} $D$ of the form of 
$\newdef{X}{\tilde{x}}{P}$ where $\fv(P)\subseteq\{ \tilde{x} \}$
and $\fn(D) = \emptyset$. 
We use $u$ to range over channel names $a$ or variables
$x$.

The language constructs are as follows: a {\em prefixed} (or guarded)
process $\pi ; P$ denotes the
behaviour $\pi$ (a {\em send action} of $\expr$ on $u$,
$\psend{u}{\expr}$, 
a {\em receive action} on $u$, bound to $y$, $\precv{u}{y}$, or an 
{\em internal action}
$\tau$) followed by process $P$; a {\em close process} 
$\close{u};P$ closes the channel
$u$ and continues as $P$; 
a {\em selection} process  
$\sel{\pi_i ; P_i}{i\in I}$ denotes a choice between
the several $P_i$ processes, where each $P_i$ is guarded by a prefix
$\pi_i$. Thus, the choice construct non-deterministically selects between any
process $P_i$ whose guarding action can be executed (note that a
$\tau$ action prefix can always be executed, cf.~\S~\ref{sect:core_sem});
the standard {\em conditional} process $\pite{\expr}{P}{Q}$, 
{\em parallel} process $P\parr Q$, and  
the {\em inactive} process $\zero$ (often omitted);
a {\em new channel} process 
$\newch{y{:}\sigma};P$ creates a new channel with payload type
$\sigma$, binding it to $y$ in the continuation $P$; 
{\em process call} 
$\defcall{X}{\tilde
  \expr,\tilde{u}}$ denotes an instance of the process definition
bound to $X$, with formal parameters instantiated to $\tilde{\expr}$
and $\tilde{u}$. 
Both {\em restriction} $(\nu c)P$ and {\em buffers} at channel $c$ 
denote \emph{runtime} constructs (i.e.~not
written explicitly by the programmer), where the
former denotes the runtime handle $c$ for a channel bound in $P$ and
a buffer for an {\em open} channel
$\buff{\tilde{\vval}}{c}{\sigma}$, containing messages $\tilde{\vval}$
of type $\sigma$, or
a buffer for a \emph{closed} channel
$\closedbuff{\tilde{\vval}}{c}{\sigma}$. A closed channel cannot be used
to send messages, but may be used for receive operations an unbounded
number of times.

Our representation of a Go program as a program $\prog$ in \fgo, written
$\pin{\{ D_i \}_{i\in I}}{P}$, consists of a set of mutually recursive
process definitions which encode all the goroutines and functions
used in the program, together with a process $P$ that encodes the
program entry point (i.e.~the \verb=main=).

\subsubsection{Example -- Prime Sieve in \fgo}\label{sect:ex_proc_sieve}
To showcase the \fgo calculus, we present a concurrent implementation of the
sieve of Eratosthenes that produces the infinite sequence of all prime
numbers. 

The implementation relies on a generator process 
$G(n,c)$ 
that
outputs natural numbers and a filter process $F(n,i,o)$ that
filters out divisible naturals. The code for the generator 
and filter
processes are given below as definitions:
\[
\begin{array}{rcl}
G(n,c) & \triangleq & \psend{c}{n};G\langle n{+1} , c\rangle\\
F(n,i,o) & \triangleq & \precv{i}{x};\pite{(x\% n \neq 0)}{\psend{o}{x};F\langle
               n,i,o \rangle\\ & &\hspace{2.6cm}}{F\langle n , i , o\rangle}
\end{array}
\]
Definition $G$ stands for the generator process: given the natural
number $n$ and channel $c$,
$G(n,c)$ sends the number $n$ along $c$ and recurses on
$n{+1}$. Definition $F$ stands for the filter: given a natural 
$n$ and a pair of channels $i$ and $o$, $F(n,i,o)$ inputs a number $x$
along $i$ and sends it on $o$ if $x$ is not divisible by $n$,
followed by a recursive call. We then need a way to chain
filters together, implementing the sieve:
\[
\begin{array}{lcl}
R(c) & \triangleq & \precv{c}{x};\newch{c'{:}\m{int}};(F\langle x
                    ,c,c'\rangle \parr R\langle c'\rangle)
\end{array}
\]
The process defined above inputs from the previous element in the
chain (either a generator or a filter), creates a new channel
which is then used to spawn a new filter process in parallel with 
a recursive call to $R$ on the new channel. Putting all the components
together we obtain the program:
\[
\pin{\{G(n,c) , F(n,i,o) , R(c) \}}
{\newch{c{:}\m{int}};(G(2,c) \parr R(c))}
\]

As we make precise
in \S~\ref{sub:fence-pred}, the execution of the processes is
\emph{fenced} insofar it is always the case that channels are
shared (finitely) by a finite number of processes. For instance, name
$c$ above is only known to $G(k,c)$ and $F(2,c,c')$. This point
is crucial to ensure the feasibility of our approach.

\subsubsection{Example -- Fibonacci in \fgo}\label{sec:sync-fib}
We implement a parallel Fibonacci number generator that computes the
$n$\textsuperscript{th} number of the Fibonacci sequence.
\[
\begin{array}{lcl}
\mathit{Fib}(n,c)  \triangleq 
\quad \pite{(n \leq 1)}{\psend{c}{n}} {\newch{c'{:}\m{int}};\\ 
\qquad (\mathit{Fib}\langle n{-1} ,
                          c'\rangle \mid \mathit{Fib}\langle n{-2} ,
                          c'\rangle \mid 
\precv{c'}{x};\precv{c'}{y};\psend{c}{x{+y}}
) }
\end{array}
\]
The definition ${Fib}(n,c)$ above tests if the given number $n$ is
less than or equal to $1$. If so, it sends $n$ on $c$ and
terminates. Otherwise, the process creates a new channel $c'$,
which is then used to run two parallel copies of $\mathit{Fib}$ for
the two predecessors of $n$. The parallel instances are composed
with inputs on $c'$ twice which are then added and
sent along $c$. 
A sample program that produces the $10$\textsuperscript{th} element of the Fibonacci
sequence is given below:
\[
\pin{\{\mathit{Fib}(n,c) \}}
{\newch{c{:}\m{int}};(\mathit{Fib}\langle 10,c \rangle \parr \precv{c}{u};\zero )}
\]
On the other hand, the following program should be deemed not live since
the outputs from  the recursive calls are never sent to the initial
$\mathit{Fib}_{\mathit{bad}}$ call and so the answer is never returned
to the main process (i.e. $\precv{c}{u};\zero$ can never fire).
\[
\begin{array}{lcl}
\mathit{Fib}_{\mathit{bad}}(n,c)  \triangleq {\newch{c'{:}\m{int}}};\\
\ (\mathit{Fib}_{\mathit{bad}}\langle n{-1},
              c'\rangle \mid \mathit{Fib}_{\mathit{bad}}\langle n{-2} ,
             c'\rangle \mid 
\precv{c'}{x};\precv{c'}{y};\psend{c}{x{+y}}
) 
\end{array}
\]

\subsection{Operational Semantics}\label{sect:core_sem}
The semantics of \fgo, written $P \tra{} Q$, is defined by
the reduction rules of Figure~\ref{fig:redsem}, together with the
standard structural congruence $P\equiv Q$ (which includes
$\equiv_\alpha$). 
For process definitions,
we implicitly assume the existence of an ambient set of definitions
$\{D_i \}_{i\in I}$.  Our semantics follows closely the semantics of
the Go language: a channel is implemented at runtime by a buffer that
is open or closed. Once a channel is closed, it may not be closed
again nor can it be used for output. However, a closed channel can always
be the subject of an input action, where the received value 
is a bottom element of the corresponding payload data
type. For now, we impose a synchronous semantics (represented in Go
with channel of size $0$), see \S~\ref{sec:asynchrony} for
the asynchronous semantics.

Rule \rulename{scom} specifies a synchronisation between a
send and a receive. Rule \rulename{sclose} defines inputs from
closed channels, which according to the semantics of Go are always
enabled, entailing the reception of a base value $\vval^\sigma$ of type
$\sigma$. Rule \rulename{close} changes the state of a buffer from
open ($\buff{\tilde{\vval}}{c}{\sigma}$) to closed
($\closedbuff{\tilde{\vval}}{c}{\sigma}$). Rule \rulename{newc} creates
a fresh channel $c$, instantiating it accordingly in the continuation
process $P$ and creating the buffer for the channel. Rule \rulename{sel}
encodes a mixed non-deterministic choice, insofar as any
subprocess $P_j$ that can exhibit a reduction may trigger the choice.
The rule \rulename{def} replaces $X$ by the corresponding process
definition (according to the underlying definition environment),
instantiating the parameters accordingly. The remaining rules are
standard from process calculus literature \cite{SangiorgiD:picatomp}.

\begin{figure}
\[
\begin{array}{l}
\rulename{scom}
\inferrule
 {\expr\conv \vval} 
 {\psend{c}{\expr};P \parr \precv{c}{y};Q \parr \buff{\emptyset}{c}{\sigma}
 \tra{}
  P \parr Q\subs{\vval}{y} \parr \buff{\emptyset}{c}{\sigma} }
\\[5mm]
\multicolumn{1}{l}{
\hspace{-2mm}
\begin{array}{ll}
\rulename{sclose}&
 {\precv{c}{y};P \parr \closedbuff{\emptyset}{c}{\sigma}}\tra{}
 {P\subs{\vval^\sigma}{y} \parr \closedbuff{\emptyset}{c}{\sigma}}
\hspace{1cm}
\\[1mm]
          \rulename{close}& 
 {\close c ; P \parr \buff{\tilde{\vval}}{c}{\sigma}}\tra{}
 { P \parr \closedbuff{\tilde{\vval}}{c}{\sigma}}\quad\quad
\\[1mm]
\rulename{tau}&
    {}
    {\typesem{\closedchans}{\tau; P}}\tra{}
    {\typesem{\closedchans}{P}}\\[1mm]
\end{array}
}
\\[1mm]
\rulename{newc}
\inferrule
  {c \notin \fn(P)}
  {{\typesem{\closedchans}{\newch{y{:}\sigma};P}}\tra{}
  {\typesem{\closedchans}{(\nu c) (P\subs{c}{y}} \mid \buff{\emptyset}{c}{\sigma})}}
  \\[5mm]
                              \multicolumn{1}{l}{
\hspace{-2mm}
\begin{array}{ll}
\rulename{par}
    \inferrule    {\typesem{N}{P} \tra{} \typesem{N'}{P'}}
    {{\typesem{N}{P \parr Q}}\tra{}
    {\typesem{N'}{P' \parr Q}}}
&
\rulename{res}
    \inferrule    {\typesem{N}{P} \tra{} \typesem{N'}{P'}}
    {{\typesem{N}{(\nu c) P}}\tra{}
    {\typesem{N'}{(\nu c) P'}}}
  \\[5mm]
\rulename{str}
\inferrule
    {P \equiv Q \tra{} Q'\equiv P'}
    {{\typesem{N}{P}}\tra{}
    {\typesem{N'}{P'}}}
    &
\rulename{sel}
\inferrule{\typesem{\closedchans}{\pi_j; P_j\parr P} \tra{} \typesem{\closedchans'}{R} \quad j \in I}
    {{\typesem{\closedchans}{\sel{\pi_i; P_i}{i\in I}}\parr P}\tra{}
    {\typesem{\closedchans'}{R}}}
  \\[5mm]
\rulename{ift}
\inferrule
    {e\conv \true}
    {{\typesem{\closedchans}{\pite{e}{P}{Q}}}\tra{}
    {\typesem{\closedchans}{P}}}
&
\rulename{iff}
\inferrule
    {e\conv \false}
    {{\typesem{\closedchans}{\pite{e}{P}{Q}}}\tra{}
    {\typesem{\closedchans}{Q}}}\\[5mm]
\end{array}
}
\\[5mm]
\rulename{def}
 \inferrule{\typesem{\closedchans}{P\subs{\tilde{\vval},\tilde{c}}{\tilde{x}} \parr Q} \rightarrow 
\typesem{\closedchans'}{R} \quad e_i \conv \vval_i \quad X(\tilde{x})
  = P\in \{ D_i \}_{i\in I}}   {{\typesem{\closedchans}{X(\tilde{e},\tilde{c})\parr Q}}{\typesem{\closedchans'}\tra{}{R}}}
\end{array}
\]
{\bf Structural Congruence}
\[
\begin{array}{l}
P \parr Q \equiv Q \parr P \qquad P \parr (Q \parr R) \equiv (P \parr
  Q) \parr R\quad 
P \parr \zero \equiv P \\
(\nub c)(\nub d)P \equiv (\nub d)(\nub c)P
\quad \;
(\nub c)\zero \equiv \zero
\quad \;
(\nub c)\buff{\tilde{\vval}}{c}{\sigma} \equiv \zero 
\\
P \mid (\nub c)Q \equiv (\nub c)(P\mid Q) \ \mbox{\scriptsize $(c\not\in \fn(P))$}
\quad 
(\nub c) \closedbuff{\tilde{\vval}}{c}{\sigma} \equiv \zero
\end{array}
\]
\caption{Reduction Semantics.}\label{fig:redsem}
\end{figure}

\subsection{Liveness and Channel Safety}\label{sub:live-safe-process}

We define a notion of liveness and channel safety for programs
through barbs in
processes (Definition~\ref{def:barbs}). Liveness identifies the
ability of communication actions to always eventually fire. 
Channel safety pertains to the semantics of channels in Go, where closing a
channel more than once or sending a message on a closed channel
raises a runtime error. 

A common pattern in the usage of select in Go is to 
introduce a timeout (or default) branch, which we model as 
a $\tau$-guarded branch. This notion of timeout  
makes the definition of liveness slightly challenging.
Consider the following:
\[
\begin{array}{ll}
P_1 \triangleq \sel{\psend{a}{v}, \, \precv{b}{x};\zero \, ,
  \, \tau ; P_{t}}{} & R_1 \triangleq \precv{a}{y};\zero \\
P_2 \triangleq \sel{\psend{a}{v}, \, \precv{b}{x};\zero}{}
&
R_2 \triangleq \precv{c}{y};\zero
\end{array}
\]
A select with a branch guarded by $\tau$ contains a branch that is
always enabled by default (since $\tau$ actions can always fire
silently). Hence if the continuation of the $\tau$ prefix $P_t$ is
live, then $P_1$ is live.  On the other hand, $P_2$ by itself
\emph{cannot} be live.  For $P_2$ to be live, it must be composed with
a process that can offer an input on $a$ or an output on $b$, with
respective live continuations. Hence $P_2\parr R_1$ is live.  However,
$P_1\parr R_i$ is not live unless $P_t \parr R_i$ is live ($i \in
\{1,2\}$).

Accounting for these features, we formalise 
safety and liveness properties, extending 
the notion of barbed process predicates \cite{MiSa92}. 
Most of the definitions are given in a standard way, with some
specifics due to the ability to close channels: 
input
barbs $P\barb{\labtyrcv{\tychana}}$, denoting that process $P$ is ready
to perform an input action on the free channel name $a$; output barbs
$P\barb{\labtysnd{\tychana}}$ are dual;
a synchronisation barb $P\barb{\labsync{\tychana}}$,
indicating that $P$ can perform a synchronisation on $a$; 
a channel close barb $P\barb{\labclose{\tychana}}$, denoting that $P$
can close channel $a$; and $P\barb{\labclsnd{\tychana}}$, denoting
that $P$ may send from closed channel $a$. 
We highlight the predicate $P\barb{\tilde{o}}$, where $\tilde{o}$ is a
set of barbs, which 
applies only for the
select construct, stating that the barbs of $\sel{\pi_i ; P_i}{i\in
  I}$ are those of all the processes that make up the external choice,
provided that \emph{all of them} can exhibit a barb. 

\begin{definition}[Barbs]~\label{def:barbs}\rm
We define the predicates
$\pi\barb{o}$, 
$P\barb{o}$ and $P\barb{\tilde{o}}$ 
with $o,o_i \in \{
\labtyrcv{\tychana}, 
\labtysnd{\tychana}, 
\labsync{\tychana},  
\labclose{\tychana}, 
\labclsnd{\tychana}\}$.
\[
\precv{c}{x}\barb{c}
\quad 
\psend{c}{\expr}\barb{\ov{c}}
\quad
\inferrule
{\pi \barb{o}}
{\pi ; Q \barb{o}}
\quad 
\close{c};Q\barb{\clbarb{c}}
\
\closedbuff{\vec{v}}{c}{\sigma}\barb{\cclbarb{c}}
\]
\[
\inferrule
{P\barb{o}}
{P\parr Q \barb{o}}
\quad 
\inferrule
  {P\barb{o}\quad a \not\in \fn(o)}
  {(\nu a)P\barb{o}}
\quad 
\inferrule
  {P\barb{o}\quad P \equiv Q}
  {Q\barb{o}}
\]
\[
\inferrule
{
  Q\subs{\tilde \expr,\tilde{a}}{\tilde{x}} \barb{o}\quad  \newdef{X}{\tilde{x}}{Q}
}
{
  \defcall{X}{\tilde{\expr}, \tilde{a}} \barb{o}
}
\]
\[
\inferrule
  {\forall i \in \{1,..,n\} \qst \pi_i \barb{o_i} }
  {\sel{\pi_i; P_i}{i\in \{1,..,n\}}\barb{ \{ o_1 \ldots o_n \} }}\quad 
\inferrule
{P\barb{a}\quad Q\barb{\ov{a}}\mbox{ or }Q\barb{\cclbarb{a}}}
{P\parr Q\barb{[a]}}
\]
\[
\inferrule
{P\barb{a}\quad \pi_i\barb{\ov{a}} } {P\parr \sel{\pi_i ; Q_i}{i \in I}\barb{[a]}}
\quad 
\inferrule
{P\barb{\ov{a}}\mbox{ or }P\barb{\cclbarb{a}}\quad \pi_i\barb{a}}
{P\parr \sel{\pi_i ; Q_i}{i \in I} \barb{[a]}}
\]
\noindent ${P}\wbarb{o}$ if ${P}\tra{}^* {P'}$ and 
${P'}\barb{o}$ with 
$o \in \{c , \ov{c} , [c], \clbarb{c}, \cclbarb{c}\}$.
\end{definition}
For example, 
we have that $\neg (P_1\barb{a})$ for any $a$,  
whereas $P_2 \barb{\ov{a},b}$. Note that
$(P_1 \mid R_1)\barb{[a]}$ and $(P_2 \mid R_1)\barb{[a]}$; and  
if $P_t\barb{o}$ then $P_1\wbarb{o}$ and 
if $P_t\barb{\ov{c}}$ then $P_1\parr R_2 \wbarb{[c]}$.

We may now define liveness and channel safety: 
a program is live if, for all the reachable process states, (a) if the
state can perform an input or output action on a channel, the
state can also eventually perform a synchronisation on that channel;
and, (b) if a state can perform a set of actions (i.e. a select where
all its guards are non-$\tau$), the state can also eventually
synchronise on one of the action prefixes.

\begin{definition}[Liveness]~\label{def:plive}\rm
The program $\prog$ satisfies liveness if for all $Q$ such that
$\prog \tra{}^* (\nub \tilde{c}){Q}$:
\begin{itemize}
\item[(a)] If ${Q}\barb{a}$ or 
${Q}\barb{\ov{a}}$ 
then 
${Q}\wbarb{\labsync{a}}$.  
\item[(b)] If ${Q}\barb{\tilde{a}}$ 
then ${Q}\wbarb{\labsync{a_i}}$ for some $a_i\in \{\tilde{a}\}$.   
\end{itemize}
\end{definition}

Channel safety states that in all reachable program
states, channels are closed at most once and no process performs
outputs on closed channels, as specified by the semantics of the Go language.

\begin{definition}[Channel Safety]~\label{def:psafe}
\rm
The program $\prog$ 
is {\em channel safe} if for all $Q$ such that 
$\prog\tra{}^*(\nub \tilde{c}){Q}$, if 
$Q\barb{\cclbarb{a}}$ 
then 
$\neg(Q\wbarb{\clbarb{a}})$ and $\neg (Q\wbarb{\ov{a}})$.
\end{definition}

 \section{ A Behavioural Typing System for \fgo}\label{sect:core_typ}
Go's channel types are related to those of the
$\pi$-calculus, where the type of a channel carries the type of the
objects that threads can send and receive along the channel. 
Our typing system augments Go's channel types by also
serving as a behavioural abstraction of a valid \fgo program,
where types take the form of CCS processes with name creation.

\subsection{Syntax of Types}
The syntax of types $T,S$ is
given in Figure~\ref{fig:types}, mirroring closely that of \fgo processes:
The type $\typrefix ;T$ denotes an output $\tysend{\tychavar}$,
input $\tyrecv{\tychavar}$ along channel $\tychavar$, or an explicit $\tau$
action (often used to encode timeouts in external choices), followed
by the behaviour denoted by type $T$. The type
$\tch{}{T_i}{i\in I}$ represents an internal choice between the
behaviours $T_i$, whereas $\tbr{}{  \typrefix_i ; T_i}{i\in I}$
denotes an external choice between behaviours $T_i$, respectively
guarded by prefixes $\typrefix_i$ which drive the choice. Types
include parallel composition of behaviours $T \parr S$, 
inaction $\zero$ and  
channel
creation $\tynew{\tychana} T$ (binding $\tychana$ in $T$). 
The type $\End[{\tychavar}];T$ denotes the closing of channel
$\tychavar$ followed by the behaviour $T$.
The type variable  
$\typevara_{X}\ENCan{\tilde{u}}$ 
associated to a process variable $X$,   
denotes the behaviour bound to
variable $\typevara_{X}$ in the definition environment, with formal
parameters instantiated to $\tilde{u}$. 
The type $\{\typevara_{X_i}(\tilde{y}_i) = T_i\}_i \ \mathsf{in} \ S$
codifies a set of (parameterised) mutually recursive type definitions
$\typevara_{X_i}(\tilde{y}_i) = T_i$, bound in $S$. This
set of equations denoted by $\EqTypes$
is the type assigned to top-level programs $\prog$. 

The type constructs $(\nu \tychana) T$, $\tyopenbuffer{\tychana}$ and
$\tyclosedbuffer{\tychana}$ denote the type representations of runtime
channel bindings, open and closed buffers, respectively. We write
$\fn(T)$ and $\fv(T)$ for the free names and variables of type $T$,
respectively; and $\allnames(T)$ denotes the set of bound and free
names.

\begin{figure}[t]
\[
\begin{array}{rcl}
T,S & :=& \typrefix ;T
    \mid \tch{}{T_i}{i\in I}
\mid \tbr{}{  \typrefix_i ; T_i}{i\in I}
\mid (T \parr S)
\mid \zero\\
& \mid & 
\tynew{\tychana} T \mid 
\End[{\tychavar}];T
\mid \typevara_{X}\ENCan{\tilde{u}}
 \mid (\nu \tychana) T
\mid  \tyopenbuffer{\tychana}
\mid \tyclosedbuffer{\tychana}\\[1mm]
\EqTypes & := &   \{\typevara_{X_i}(\tilde{y}_i) = T_i\}_i \ \mathsf{in} \
          S   \qquad \typrefix \coloneqq  
  \tysend{\tychavar}
   \mid  
  \tyrecv{\tychavar}
   \mid  
  \tau
\end{array}
\]
\caption{Syntax of Types.}\label{fig:types}
\end{figure}

\subsection{Typing System}
We first explain the two essential differences from (linear or
session-based) type systems of the
$\pi$-calculus~\cite{THK,DBLP:journals/tcs/IgarashiK04,Huttel:2016}:

{\bf Sharing of Channels.} We do not enforce linear (disjoint) 
channel usages, allowing processes to have races.  
For instance, the process below (with shared $y$) 
is typable:
\[\newch{y{:}\m{bool}};(\psend{y}{\true};\zero \parr
\psend{y}{\false};\zero\parr\precv{y}{x};\zero)\] 

{\bf Conditionals.} We do not enforce the same types of both branches 
of the conditional. 
This design choice stems from the fact
that most real programs make use of conditionals precisely to identify
points where behaviours need to be different. For instance, consider
the following definition: 
\[ {X}(c) = \precv{c}{x};\pite{x\geq
  0}{{X}\langle c\rangle}{\zero}\]
The recursive process defined by
${X}(c)$ receives a potentially
unbounded number of positive integers, stopping when the
received value $x$ is less than $0$. 
To type such a commonplace programming pattern, 
we allow the branches in the conditional to hold different types
(which are also incompatible by the usual branch subtyping). 

\paragraph{Process Typing}
The judgement ($\G \vdash P \ts T$) for processes is defined in
Figure~\ref{fig:typing} 
where $\G$
is a typing environment that maintains information about channel
payload types, types of bound communication variables and recursion
variables, $P$ is a process and $T$ a behavioural type. 

We write $\G \vdash \mathcal{J}$ for $\mathcal{J} \in \G$ and
$\G \vdash e : \sigma$ to state that the expression $e$ is
well-typed according to the types of variables in $\G$. We write
$u{:}\m{ch}(\sigma)$ to denote that $u$ stands for a channel with
payload type $\sigma$. We omit the
typing rules of expressions $e$, given that
expressions only include basic data types. We write $\domain{\G}$ for
the set of channel bindings in $\G$.  

The rules implement a very close
correspondence between processes and their respective types: Rule
{\trulename{out}} types output processes with the output prefix type
$\tsend{u}{\sigma};T$, checking that the type of the object to be sent
matches the payload type $\sigma$ of channel $u$, and that the
continuation $P$ has type $T$. The rule \trulename{in} for inputs
is dual. 
Rule \trulename{sel} types the select construct with the external
choice type, whereas rule {\trulename{if}} types the conditional
as a binary internal choice between the type $S$ corresponding to
$P$ and the type $T$ corresponding to $Q$. 

The typing rules for close, zero, parallel and $\tau$ are straightforward. 
Rule {\trulename{new}} allocates a fresh type-level channel name
with payload type $\sigma$. 
Rule {\trulename{var}} matches a
process variable with its corresponding type variable, checking that
the specified arguments have the appropriate types. 

\paragraph{Program Typing}
The judgement 
($\G \vdash \prog \ts \EqTypes$) is defined in
Figure~\ref{fig:typing}. 
A process declaration
$X(\tilde{x}{:}\tilde{\sigma},\tilde{y}{:}\m{ch}(\tilde{\sigma}')) =
P$ is matched with $\typevara_X(\tilde{y}) = T$, connecting the process level
variable $X$ with the type variable $\typevara$, where $P$ may use any of the
parameters specified in the recursion variable. The typing rule for
programs \trulename{def} assigns a program the type 
$\{\typevara_{X_i}(\tilde{y}_i) = T_i\}_{i \in I}\ \mathsf{in} \ S$,
checking that each definition is typed with
$\typevara_{X_i}(\tilde{y}_i) = T_i$ and that the main process $Q$
has type $S$. 

\paragraph{Runtime Process Typing}
The judgement ($\G \vdash_B P \ts T$) 
types a process created after execution of
a program (called runtime process). 
$B$ is a set of channels with associated runtime buffers
to ensure their uniqueness. 
Runtime channel bindings are typed by rule {\trulename{res}}, 
where given a process of type $T$ that can use the buffered
channel $c$, we type $(\nub c)P$ with $(\nub c)T$ removing $c$ from
the set $s$ since it is local to $P$ (and $T$).  
Closed and open buffers are typed by rules 
{\trulename{cbuff}}
and {\trulename{buff}}, respectively, noting that the set of active
buffers is a singleton containing the appropriate buffer
reference. The parallel rule \trulename{parr} 
ensures that the buffers of both processes do not overlap (hence 
only a single buffer for each name exists in the context). 

\begin{figure}
\framebox{$\G \vdash P \ts T$}\\
$
\begin{array}{c}
\trulename{out}
\inferrule{\G \vdash u {:} \m{ch}(\sigma) \quad \G \vdash e : \sigma \quad  \G \vdash P \ts  T }
{\G \vdash \psend{u}{e};P \ts \tsend{u}{\sigma};T  }
\\[1.5em]
\trulename{in}
\inferrule{\G \vdash u {:} \m{ch}(\sigma) \quad \G , x{:}\sigma\vdash P \ts T }
{\G \vdash \precv{u}{x};P \ts \trecv{u}{\sigma};T}
\ 
\trulename{tau}
\inferrule{\G \vdash P \ts T}
{\G \vdash \tau ; P \ts  \tau ; T}\\[1.5em]
\trulename{close}
\inferrule{ \G \vdash P \ts T }
{\G \vdash \close{u};P \ts \End[u];T 
   }
\quad 
\trulename{zero}
\inferrule{ }
{\G \vdash \zero \ts \zero }
\\[1.5em]
\trulename{sel}
\inferrule{ \G \vdash \pi_i;P_i \ts \kappa_i;T_i }
{\G \vdash \sel{\pi_i; P_i}{i\in I} \ts \tbr{}{\kappa_i;T_i}{i\in I} }
\\[1.5em]
\trulename{if}
\inferrule{\G \vdash e : \m{bool} \quad  \G \vdash P \ts S  \quad
  \G \vdash Q \ts T  }
{\G \vdash \pite{e}{P}{Q} \ts \tch{}{S \,
    , \, 
    T}{} }\\[1.5em]
\trulename{new}
\inferrule{\G , y{:}\m{ch}(\sigma) \vdash P \ts T \quad c\not\in \domain{\G}   \cup \fn(T)}{\G \vdash {\newch{y{:}\sigma}};P \ts
  \tynew{c}T\subs{c}{y} }
\\[1.5em] 
\trulename{par}
\inferrule{\G \vdash P \ts T \quad
 \G \vdash Q \ts S}
{\G \vdash P \parr Q \ts (T \parr S)}\\[1.5em] 
\trulename{var}
\inferrule{\G \vdash \tilde{e} {:} \tilde{\sigma}\quad  
\G \vdash \tilde{u} {:}\mathsf{ch}(\tilde{\sigma}')}
{\G, X(\tilde{\sigma},\mathsf{ch}(\tilde{\sigma}'))\vdash X\langle\tilde{e} , \tilde{u}\rangle \ts \typevara_X\langle \tilde{u}\rangle }\\[1.5em]
\end{array}$\\[1mm]
\framebox{$\G \vdash \prog \ts \EqTypes$}\\
\hspace*{-3mm}$
\begin{array}{c}
\trulename{def}
\inferrule{\begin{array}{c}
\G , {X_i(\tilde{\sigma}_i,\m{ch}(\tilde{\sigma}'_i))},
\tilde{x}_i{:}\tilde{\sigma}_i,\tilde{y}_i{:}\m{ch}(\tilde{\sigma}'_i)
\vdash P_i
\ts
T_i
\\
\G, X_1(\tilde{\sigma}_1,\m{ch}(\tilde{\sigma}_1')),...,
X_n(\tilde{\sigma}_n,\m{ch}(\tilde{\sigma}_n'))
 \vdash Q \ts S
\end{array}
}
{\G \vdash
\{
X_i(\tilde{x_i},\tilde{y_i}) = P_i 
\}_{i\in I} 
\ \mathsf{in} \ Q\ts 
\{\typevara_{X_i}(\tilde{y}_i) = T_i\}_{i \in I} \ \mathsf{in} \ S}\\[1em]
\end{array}$\\[1mm]
\framebox{$\G \vdash_B P \ts T$}\\
$\begin{array}{c}
\trulename{int}
\inferrule{\G \vdash P \ts T}
{\G \vdash_\emptyset P \ts T}
\quad 
\trulename{res}
\inferrule{\G , c{:}\m{ch}(\sigma) \vdash_B P \ts T}
{\G \vdash_{B\setminus c} (\nub c)P \ts (\nub c)T}\\[1.5em]
\trulename{cbuff}
\inferrule{\G \vdash a{:}\m{ch}(\sigma)}
{\G \vdash_{\{a\}}\closedbuff{\tilde{\vval}}{a}{\sigma}  \ts
  \tyclosedbuffer{\tychana} }\quad  
\trulename{buff}
\inferrule{\G \vdash a{:}\m{ch}(\sigma)}
{\G
\vdash_{\{a\}}\buff{\tilde{\vval}}{a}{\sigma}
  \ts \tyopenbuffer{\tychana} }\\[1.5em]
\trulename{parr}
\inferrule{\G \vdash_{B} P \ts T \quad
 \G \vdash_{B'} Q \ts S \quad B \cap B' = \emptyset}
{\G \vdash_{B\cup B'} P \parr Q \ts (T \parr S)}
\end{array}
$
\caption{Typing Rules (Processes and Programs).~\label{fig:typing}}
\end{figure}

\begin{notation}\rm
In the remainder of this paper, 
we refer to the type of a
program as a system of type equations $\EqTypes$ which is obtained
from the program by collecting all the types for definitions and
adding a distinguished unique definition $\typevara_0() = 
S$ for the program entry point. We often write $X_0$ to stand for the
process variable of the program entry point.
\end{notation}

 \section{Bounded Verification of Behavioural Types}
\label{sec:types}
This section introduces our main definition, {\em fencing}, and a bounded
verification of liveness and channel safety for types. Our development
consists of the following steps:

\begin{description}
\item[Step 1.] Define a syntactic restriction on types, dubbed
  a \emph{fence}, guaranteeing that whenever \emph{fenced} types model
  a program that spawns infinitely many processes, the program actually
  consists of finitely many communication patterns (which may be
  repeated infinitely many times). 
  \item[Step 2.] Define a \emph{symbolic semantics} for
  types, which generates a representative \emph{finite-state}
  labelled transition system (LTS) whenever types validate the
  fencing predicate.
  \item[Step 3.] Prove that liveness and channel safety are decidable
  for the bounded \emph{symbolic executions} of fenced types.
\end{description}

\subsection{Types as Processes: Semantics}\label{sub:sync-ty-sem}
\begin{figure}
  \[
  \begin{array}{c}
        \ltsrulename{snd}\quad  
    {\tsend{\tychana}{\sigma};T}\tra{\labtysnd{\tychana}}{T}
        \quad
    \ltsrulename{rcv}\quad  
    {\trecv{\tychana}{\sigma};T}\tra{\labtyrcv{\tychana}}{T}
\quad 
    \ltsrulename{tau}\quad 
    {\tau ; T}\tra{\tau}{T}
\\[0.5em]
        \ltsrulename{sel}
    \inferrule{j \in I}
    {\tch{}{T_i}{i\in I}\tra{\tau}T_j}
    \qquad
    \ltsrulename{bra}
    \inferrule{\typesemF{\typrefix_j ;T_j}  \semty{\labcom}
      \typesemFdash{T_j}}
    {\tbr{}{ \typrefix_i ; T_i}{i\in I}\tra{\labcom}T_j}
    \\[1.5em]
    \ltsrulename{par}
    \inferrule
    {\typesemF{T} \semty{\labcom} \typesemFdash{T'}}
    {{T \parr S}\tra{\labcom}{T'\parr S}}
    \quad
    \ltsrulename{com}
    \inferrule{\typesemFone{T} \semty{\beta} \typesemFone{T'} 
      \quad
      \typesemFtwo{S} \semty{\labtyrcv{\tychana}} \typesemFtwo{S'}
      \quad
      \beta = {\labtysnd{\tychana}}, {\labclsnd{\tychana}}
    }
    {{T \parr S}\tra{\labsync{\tychana}}
    {T'\parr S'}}
        \\[0.5em]
    \ltsrulename{new}\ \,
    {\tynew{\tychana}T}\tra{\tau}
    {(\nu \tychana)(T \parr {\tyopenbuffer{\tychana}})}
\quad 
\ltsrulename{end}\ \,
    {}
    {\End[\tychana];T}\tra{\labclose{\tychana}}{T}
    \\[0.5em]
\begin{array}{l}
\ltsrulename{buf}\quad 
    {}
    {\tyopenbuffer{\tychana}}\tra{\labclosedual{\tychana}}{\tyclosedbuffer{\tychana}}
     \\[0.1em]
    \ltsrulename{cld}\quad 
    {}
    {\tyclosedbuffer{\tychana}}\tra{\labclsnd{\tychana}}{\tyclosedbuffer{\tychana}}
\end{array}
\quad 
\ltsrulename{close}
\inferrule{\typesemFone{T} \semty{\labclose{\tychana}} \typesemFone{T'} 
      \quad
      \typesemFtwo{S} \semty{\labclosedual{\tychana}} \typesemFtwo{S'}
          }
    {{T \parr S}\tra{\tau}
    {T'\parr S'}}
    \\[0.5em]
        \ltsrulename{res-1}
    \inferrule
    {
      \typesemF{T} \semty{\labcom} \typesemFdash{T'} 
      \quad
      \freenames{\labcom}  \neq \{\tychana\} 
    }
    {(\nu \tychana) T\tra{\labcom}(\nu \tychana)T'}
    \qquad
    \ltsrulename{res-2}
    \inferrule{\typesemF{T} \semty{\labsync{\tychana}} \typesemF{T'}}
    {(\nu \tychana) T\tra{\tau}(\nu \tychana)T'}
    \\[1.5em]
   \ltsrulename{eq}
    \inferrule
    {T \equiv_{\alpha} T' \quad 
      \typesemF{T} \semty{\labcom} \typesemFdash{T''}}
    {T'\tra{\labcom}T''}
    \qquad
       \ltsrulename{def}
    \inferrule
    {
      \typesemF{T} \subs{\tilde{a}}{\tilde{\tyvar}} \semty{\labcom}\typesemFdash{T'}
      \quad
      \typevara(\tilde{\tyvar}) = T 
    }
    {\tdefcall{\typevara}{\tilde{a}}\tra{\labcom}{T'}}
  \end{array}
  \]
\caption{LTS Semantics of Types.}\label{fig:types-sem}
\end{figure}
 The semantics of our types is given by the
labelled transition system (LTS), extending that of CCS,  
defined in Figure~\ref{fig:types-sem}. 
The labels, ranged over by $\alpha$ and $\beta$, have the form:
\[
\alpha,\beta \coloneqq \ \labtysnd{\tychana} \ | \ \labtyrcv{\tychana}\ | \ 
\tau \ | \ \labsync{\tychana} 
\ | \ \labclose{\tychana}
\ | \ \labclosedual{\tychana}
\ | \ \labclsnd{\tychana}
\]
Labels denote send and receive actions ($\labtysnd{\tychana}$ and
$\labtyrcv{\tychana}$), silent transitions $\tau$,
synchronisation over a channel $\labsync{\tychana}$, the request and
acceptance of channel closure ($\labclose{\tychana}$ and
$\labclosedual{\tychana}$), and send actions from a closed channel
$\labclsnd{\tychana}$. 
We write $\typevara(\tilde{\tyvar}) = {T}$ if 
$\typevara(\tilde{\tyvar}) = {T}$ is in $\EqTypes$.

Rule \ltsrulename{snd} (resp.\ \ltsrulename{rcv}) allows a type to
emit a send (resp.\ receive) action on a channel $\tychana$.
Rules \ltsrulename{sel} and \ltsrulename{bra} model internal and
(mixed) external choices, respectively.
Rule \ltsrulename{com} allows two types to synchronise on a dual action
(send with receive, or closed channel with receive).
Rule \ltsrulename{new} creates a new (open) channel for $\tychana$.
Rule \ltsrulename{end}, together with rule \ltsrulename{close}, allows a
type to request the closure of channel $\tychana$.
Rule \ltsrulename{buf} models a transition from an open channel to a
closed one, and rule \ltsrulename{cld} models the perpetual ability of a
closed channel to emit send actions.
The other rules are standard from CCS. In Figure~\ref{fig:types-sem}, we omit the symmetric rules for
\ltsrulename{close}, \ltsrulename{par}, and \ltsrulename{com}.  
We define structural congruence rules over types as follows:
\[
\begin{array}{l}
  T \parr S \equiv S \parr T
  \quad
  T \parr (T' \parr S) \equiv (T \parr T') \parr S
  \quad 
  T \parr \zero \equiv T 
  \\
  (\nub \tychana)(\nub \tychanb)T \equiv (\nub \tychanb)(\nub \tychana)T
  \ \
  (\nub \tychana) \zero \equiv \zero 
  \ \
  (\nub \tychana) {\tyclosedbuffer{\tychana}} \equiv \zero
\ \  (\nub \tychana) \tyopenbuffer{\tychana} \equiv \zero 
\\
  T \mid (\nub \tychana)S \equiv (\nub \tychana)(T\mid S) 
  \  \mbox{\scriptsize $(\tychana\notin \fn(T))$}
\ \ \ 
  T \equiv_\alpha T' \ \Rightarrow T \ \equiv T'
\end{array}
\]
We write $\tra{}$ for $\tra{\tau}\cup\equiv$ and 
$T \tra{}^\ast \tra{\alpha}$ if there exist $T'$ and $T''$
such that $T\tra{}^\ast T'\tra{\alpha}T''$.

\subsection{Fenced Types and Fencing Predicate}\label{sub:fence-pred}

This section develops {\bf Step 1} by defining the fencing
predicate. We illustrate the key intuitions with an example.

Recall 
the prime sieve program
(\S~\ref{sect:ex_proc_sieve}) given in \S~\ref{sec:overview}, as
inferred by the type system of \S~\ref{sect:core_typ}.
We represent the execution wrt.\ the
semantics of \S~\ref{sub:sync-ty-sem}, via the diagram below.
\[
\begin{tikzpicture}[node distance= 0.5cm and 0.2cm
  , scale=0.47, every node/.style={transform shape}
  , font=\LARGE
  ]
    \node (g) { $\primeG{\color{\primecolorA}{\primeN{0}}}$};
    \node[gray,left=of g,xshift=-0.5cm] (t0) {$\typevara_0 \langle \rangle$};
    \node[gray,below = of g] (r1) { $\primeR{\color{\primecolorA}{\primeN{0}}}$};
  \node[right=of g] (f1) { $\primeF{\color{\primecolorA}{\primeN{0}}}{\color{\primecolorB}{\primeN{1}}}$};
    \node[gray,below = of f1] (r2) { $\primeR{\color{\primecolorB}{\primeN{1}}}$};
  \node[right=of f1] (f2) { $\primeF{\color{\primecolorB}{\primeN{1}}}{\primeN{2}}$};
  \node[gray,below = of f2] (r3) { $\primeR{\primeN{2}}$};
    \node[right = of f2,font=\small] (d1) {{$\ldots$}};
    \node[right=of d1] (fkm) { $\primeF{\primeN{k-1}}{\color{\primecolorC}{\primeN{k}}}$};
  \node[gray,below = of fkm] (rkm) { $\primeR{\color{\primecolorC}{\primeN{k}}}$};
    \node[font=\small]  (d2) at (r3-|d1)   { $\ldots$};
    \node[right=of  fkm] (fk) { $\primeF{\color{\primecolorC}{\primeN{k}}}{\primeN{k+1}}$};
  \node[below = of fk] (rk) { $\primeR{\primeN{k+1}}$};
    \node[right = of fk,font=\small] (d3) { $\ldots$};
  \node[font=\small]  (d4) at (rk-|d3)   { $\ldots$};
        \draw[-latex] (t0) -- (g);
  \draw[-latex] (t0) -- (r1);
    \draw[-latex] (r1) -- (f1);
  \draw[-latex] (r1) -- (r2);
  \draw[-latex] (r2) -- (f2);
  \draw[-latex] (r2) -- (r3);
    \draw[-latex] (rkm) -- (fk);
  \draw[-latex] (rkm) -- (rk);
      \begin{pgfonlayer}{background}
                            \path[draw=\primecolorA,fill=\primecolorA!5,fill opacity=0.3,densely dotted,rounded corners]
    (t0.north west) -- (f1.north east) -- (f1.south east)
    -- (r1.north east) 
    -- (r1.south east) 
    -- (r1.south west) 
    -- (r1.north west) 
    -- (t0.south west) -- cycle;
        \path[draw=\primecolorB,fill=\primecolorB!5,fill opacity=0.3,densely dotted,rounded corners]
    (f1.north west) -- (f2.north east) -- (f2.south east)
    -- (r2.north east) 
    -- (r2.south east) 
    -- (r2.south west)
    -- (r2.north west)
    -- (f1.south west)
    -- cycle
    ;
        \path[draw=\primecolorC,fill=\primecolorC!5,fill opacity=0.3,densely dotted,rounded corners]
    (fkm.north west) -- (fk.north east) -- (fk.south east)
    -- (rkm.north east)
    -- (rkm.south east)
    -- (rkm.south west)
    -- (rkm.north west)
    -- (fkm.south west)
    -- cycle
    ;
              \end{pgfonlayer}
  \end{tikzpicture}
\]
In the diagram, a node represents an instance of a concurrently
executing type (or thread) and the arrows represent the parent-child relation.
For instance, $\typevara_0 \langle \rangle$ is the parent of 
$\primeG{\primeN{0}}$ and $\primeR{\primeN{0}}$.

We observe that there are only finitely many kinds of
concurrent threads (i.e.\ $\typevara_0 \langle \rangle$, $\primeG{a}$,
$\primeF{a}{b}$, and $\primeR{a}$).
Also given any name in the diagram, e.g.\
$\primeN{1}$, it is shared only by finitely many threads. For instance,
$\primeN{1}$ is ``known'' only to $\primeF{\primeN{0}}{\primeN{1}}$,
$\primeF{\primeN{1}}{\primeN{2}}$, and $\primeR{\primeN{1}}$.
Note that this situation does not change during the execution of the
program because (1) types cannot exchange names in communication and
(2) the grand-children of, e.g.\ $\primeR{\primeN{1}}$, do not know $\primeN{1}$,
hence they cannot spawn further threads sharing that name. 

Intuitively, we can observe that despite the fact that the prime sieve
generates an unbounded number of fresh channels and threads, it
consists of finitely many different communication patterns
(i.e.~the coloured regions above). It is this general
pattern that we capture with \emph{fencing}.

\subsubsection{Fencing Types}
We introduce a judgement which 
ensures that a system of types is \emph{fenced}:
given a finite set of names called 
\emph{fence}, the types executing within that fence approximate 
a form of finite control.
We use judgements of the form
$
\tytyJudge{\tytyRes}{\tytyUnr}{T}
$
to guarantee that a type definition $\typevara(\tilde{x}) = T$ is
\emph{fenced}; where $\tytyProc$ records previously encountered
recursive calls, $\tytyRes$ represents the names that $\typevara$ can
use if $T$ is \emph{single-threaded}, and $\tytyUnr$ represents the
names that a sub-term of $T$ can use if $T$ is a \emph{multi-threaded}
type.
We write $\varepsilon$ for the empty environment.
\begin{figure}
\framebox{$\tytyJudge{\tytyRes}{\tytyUnr}{T}$}
  \[
  \begin{array}{c}
            \tytyRule{axiom}
    {
      {\tytyRes \neq \varepsilon \ \lor \ \tilde{\tychavar} \tytyprec \tytyUnr}
    }
    {\tytyJudge{\tytyRes}{\tytyUnr}{ {\tdefcall{\typevara}{\tilde{\tychavar}}} }}
    \qquad
    \tytyRule{def-\ensuremath{\in}}
    {
      {\tdefcall{\typevarb}{\tilde{\tychavar} } \in \tytyProc}
    }  
    {\tytyJudge{\tytyRes}{\tytyUnr}{ {\tdefcall{\typevarb}{\tilde{\tychavar}}} }}
                                                    \\[1.5em]
        \tytyRule{end}
    {\,}
    {\tytyJudgeDef{\zero}}
    \qquad
    \tytyRule{par}
    {
      \tytyJudge{\varepsilon}{\tytyUnr \cdot \tytyRes}{T_1}
      \quad
      \tytyJudge{\varepsilon}{\tytyUnr \cdot \tytyRes}{T_2}
          }
    {\tytyJudge{\tytyRes}{\tytyUnr}{T_1 \parr T_2}}
        \\[1.5em]
        \tytyRule{def-\ensuremath{\notin}}
    {
      \typevara \neq \typevarb
      \ \
      {\tdefcall{\typevarb}{\tilde{\tychavar} } \notin \tytyProc}
      \ \
      \tytyJudgeX{\tytyRes}{\tytyUnr}{\typevara}
      { T\subs{\tilde{\tychavar}}{\tilde{\tyvar}} }
      {\tytyProc \cdot \tdefcall{\typevarb}{\tilde{\tychavar} }}
      \ \
      \typevarb(\tilde{\tyvar}) = T    }  
    {\tytyJudge{\tytyRes}{\tytyUnr}{ {\tdefcall{\typevarb}{\tilde{\tychavar}}} }}
    \\[1.5em]
    \tytyRule{sel}
    {\forall i \in I \qst \tytyJudgeDef{T_i}}
    {\tytyJudgeDef{{\tch{}{T_i}{i\in I}}}}
    \quad\qquad
    \tytyRule{bra}
    {\forall i \in I \qst \tytyJudgeDef{T_i}}
    {\tytyJudgeDef{{\tbr{}{T_i}{i\in I}}}}
    \\[1.5em]
    \tytyRule{end}
    {\tytyJudgeDef{T}}
    {\tytyJudgeDef{\End[{\tychavar}]; T}}
    \
    \tytyRule{res}
    {\tytyJudgeDef{T}}
    {\tytyJudgeDef{\tynew{\tychana}T}}
    \
    \tytyRule{pref}
    {\tytyJudgeDef{T}}
    {\tytyJudgeDef{\typrefix ; T}}
                                      \end{array}
  \]
    \caption{Rules for Fencing,  
        where $\tytyenvdef$ stands for $\tytyProc \tytyenvsep \tytyRes
    \tytyenvsep \tytyUnr$.}\label{fig:fencing-pred-rule}
\end{figure}
 The judgement $\tytyJudge{\tytyRes}{\tytyUnr}{T}$ holds if it can be
inferred by the rules of Figure~\ref{fig:fencing-pred-rule}, which use
a relation on sequences of names that ensures a strictly decreasing
usage of names, defined below.
\begin{definition}[$\tytyprec$-relation] We write 
  $\tilde{u}
  \tytyprec
  \tilde{x}$ iff
  (1) $\tilde{x} =  x_{1} \cdots x_n$,
  (2) $\tilde{u} = x_{k+1} \cdots x_n \cdot \tychana_1 \cdots \tychana_k$, with $k \geq 1$, and
  (3) $\forall 1 \leq i \leq n \qst  \forall 1 \leq j \leq k \qst x_i \neq \tychana_j$.
                                  \end{definition}
\noindent The relation $\tytyprec$ enforces that types featuring parallel
composition have a ``finite memory'' wrt.\ the names over which they can
recurse. For instance, $yza \tytyprec xyz$, but $xaz \not\prec xyz$.

We comment on the key rules of
Figure~\ref{fig:fencing-pred-rule}.
Rule \frulename{par} identifies multi-threaded types by moving the
$\tytyRes$ environment to the $\tytyUnr$ environment.
The axiom \frulename{axiom} states that
${\tdefcall{\typevara}{\tilde{\tychavar}}}$ is fenced if either ($i$)
${\tytyRes \neq \varepsilon}$, i.e.\ the type corresponding to
$\typevara$ does not contain any parallel composition; or ($ii$)
${\tilde{\tychavar} \tytyprec \tytyUnr}$ holds, i.e.\ any recursive
call over $\typevara$ uses strictly fewer (non newly created) names.
The second part of the premise guarantees that after a certain number
of recursive calls, the type $\typevara$ will have completely
``forgotten'' the names it started executing with; hence moving
outside of the fence.
Rules \frulename{def-\ensuremath{\in}} and 
\frulename{def-\ensuremath{\notin}} deal with recursive calls to
different type variables.
\begin{definition}[Fenced types]\label{def:fencedtype}\rm 
We say $\EqTypes$ is {\em fenced} (denoted by 
$\tytyJudgeE{\EqTypes}$) if  
for
all $\typevara(\tilde{\tyvar})= T$ in $\EqTypes$, either
$\tilde{\tyvar} = \varepsilon$ or
$\tytyJudgeX{\tilde{\tyvar}}{\varepsilon}{\typevara}{T}{\varepsilon}$ holds.
\end{definition}

A consequence of fencing is that for each equation
$\typevara(\tilde{\tyvar})= T$ such that $\tilde{\tyvar} \neq
\varepsilon$, either
(1) $T$ is \emph{single-threaded} (i.e.\ there is no parallel
composition within $T$), hence there is no restriction as to the
parameters $\typevara$ recurses on; or
(2) $T$ is \emph{multi-threaded} (i.e.\ there is a parallel
composition within $T$), hence for each occurrence of a
$\typevara$-recursion, at least one of the parameters must be
forgotten.
In the prime sieve example, the types $\primeGN(x)$ and
$\primeFN(x,y)$ always recurse on the same parameters, but they do not
include parallel composition.
While the type $\primeRN(x)$ has a parallel composition, its
parameter $x$ is ``forgotten'' at the $\primeR{b}$ recursive call
(i.e.\ $b \tytyprec x$).

\subsubsection{Examples on Fencing}
\paragraph{No Fence} We now give an example that does \emph{not} validate
the fencing predicate.
The types below model a program that spawns a reader and
a writer on a channel $a$ infinitely many times.
\[
\begin{array}{r@{\ }lr@{\ }l}
\typevarfont{w}(x) \triangleq 
  & \tysend{x}; \typevarfont{w}\langle x \rangle
   =  T_w
  & \typevara_1(x) \triangleq & 
  \typevarfont{w}\langle x \rangle
  \parr
  \typevarfont{r}\langle x \rangle
  \parr
  \typevara_1 \langle  x\rangle
   =  T_1
  \\
  \typevarfont{r}(x) \triangleq & \tyrecv{x}; \typevarfont{r}\langle x \rangle
    =  T_r & 
  \typevara_0()   \triangleq & 
  \tynew{a}
  (
  \typevara_1 \langle a \rangle
  )
   =  T_0
  \end{array}
\]
We have:
$\tytyJudgeX{x}{\varepsilon}{\typevarfont{w}}{T_w}{\varepsilon}$,
$\tytyJudgeX{x}{\varepsilon}{ \typevarfont{r}}{T_r}{\varepsilon}$, and
$\tytyJudgeX{{\varepsilon}}{\varepsilon}{\typevara_0}{T_0}{\varepsilon}$
hold (since they do not feature any parallel composition).
However,
$\tytyJudgeX{{x}}{\varepsilon}{\typevara_1}{T_1}{\varepsilon}$ does
\emph{not} hold.
This is due to the fact that the axiom \frulename{axiom}
cannot be applied (i.e.~$\neg (
x \tytyprec x)$).

The topology induced by the type is given as:
\begin{center}
\begin{tikzpicture}[node distance= 0.7cm and 0.1cm
  , font=\Large,
  , edge from parent/.style={draw,-latex}
    , scale=0.55,transform shape
  ]
  \node (root) {$\typevara_1 \langle \nofencechan \rangle$};
    \node[left=of root,xshift=-1cm] (t0) {$\typevara_0 \langle \rangle$};
    \draw[-latex] (t0) -- (root);
    \node[below left=of root,xshift=0.5cm]  (w1)   {$\typevarfont{w}\langle \nofencechan \rangle$};
  \node[below right=of root,xshift=-0.5cm] (r1)  {$\typevarfont{r}\langle \nofencechan \rangle$};
  \draw[-latex] (root) -- (w1);
  \draw[-latex] (root) -- (r1);
    \node[right= of root,xshift=1.5cm] (root1) {$\typevara_1 \langle \nofencechan \rangle$};
  \node[below left=of root1,xshift=0.5cm] (w2)   {$\typevarfont{w}\langle \nofencechan \rangle$};
  \node[below right=of root1,xshift=-0.5cm] (r2)  {$\typevarfont{r}\langle \nofencechan \rangle$};
  \draw[-latex] (root) -- (root1);
  \draw[-latex] (root1) -- (w2);
  \draw[-latex] (root1) -- (r2);
      \node[right= of root1,xshift=2cm] (root3) {$\typevara_1 \langle \nofencechan \rangle$};
  \node[below left=of root3,xshift=0.5cm] (w3)   {$\typevarfont{w}\langle \nofencechan \rangle$};
  \node[below right=of root3,xshift=-0.5cm] (r3)  {$\typevarfont{r}\langle \nofencechan \rangle$};
  \draw[dotted] (root1) -- (root3);
  \draw[-latex] (root3) -- (w3);
  \draw[-latex] (root3) -- (r3);
    \draw[dotted] (r2) -- (w3);
        \begin{scope}[on background layer]
                               \path[draw=\fibocolorA,fill=\fibocolorA!5,fill opacity=0.7,densely dotted,rounded corners]
   (t0.north west) -- (root3.north east)
   -- (root3.south east)
   -- (r3.north east)
   -- (r3.south east)
   -- (w1.south west)
   -- (w1.north west)
   -- (t0.south west)
   -- cycle
   ;
  \end{scope}
  \end{tikzpicture}
\end{center}
We can observe that there are infinitely many instances of
types that ``know'' the name $a$.

\paragraph{Fibonacci}
Below, we give the types for the Fibonacci example, cf.~\S~\ref{sec:sync-fib}.
\[
\begin{array}{lclcr}
  \fibGN(x) & \triangleq & 
  \tysend{x}
  \ \oplus \
  \tynew{b}
    (
  \fibG{b}
  \parr
  \tyrecv{b};
  \tyrecv{b};
  \tysend{x}
  \parr
  \fibG{b}
  )
  & =   T_{\fibGN}
  \\
  \typevara_0()  & \triangleq &
  \tynew{a}
  (
  \fibG{a}
  \parr
  \tyrecv{a}
  ) 
  & =  T_0
  \end{array}
\]
Observe that $\typevara_0()$ is trivially fenced since it has no
parameter
(cf.\ Definition~\ref{def:fencedtype}). The judgement 
$\tytyJudgeX{x}{\varepsilon}{\fibGN}{T_{\fibGN}}{\varepsilon}$ 
also holds since we have $b \tytyprec x$.
Essentially, the equation $\fibGN$ validates the fencing predicate
because each recursively spawned child does not have access to the
parameter $x$.
We illustrate the behaviour of this type in the diagram below.
\begin{center}
\begin{tikzpicture}[node distance= 0.7cm and 0.3cm
  , font=\Large,
  , edge from parent/.style={draw,-latex}
  , level 1/.style={sibling distance=50mm}
  , level 1/.style={level distance=12mm},
  , level 2/.style={sibling distance=30mm}
  , level 3/.style={sibling distance=20mm}
  , level 4/.style={sibling distance=10mm}
  , scale=0.55,transform shape
  ]
  \node (root) {$\typevara_0 \langle \rangle$}
  child{ 
    node (f1) {$\fibG{\color{\fibocolorA}{\fibChan}}$}
    child{
      node (f11) {$\fibG{\color{\fibocolorB}{\fibChanB_1}}$}
      child{
        node (f111) {$\fibG{\fibChanB_{2}}$}
        \fibinftrees{a}
      }
      child{
        node (f113) {$\fibG{\fibChanB_{2}}$}
        \fibinftrees{b}
      }
      child{
        node (b1c1a) {$
          \tyrecv{\fibChanB_{2}};
          \tyrecv{\fibChanB_{2}};
          \tysend{\color{\fibocolorB}{\fibChanB_1}}$}
      }
          }
    child{
      node (f12) {$
        \tyrecv{{\color{\fibocolorB}{\fibChanB_1}}};
        \tyrecv{\color{\fibocolorC}{\fibChanB_1}};
        \tysend{\color{\fibocolorA}{\fibChan}}$}
    }
    child{
      node (f13) {$\fibG{\color{\fibocolorC}{\fibChanB_1}}$}
      child{
        node (f122) {$
          \tyrecv{\fibChanB_{3}};
          \tyrecv{\fibChanB_{3}};
          \tysend{\color{\fibocolorC}{\fibChanB_1}}$}
      }
      child{
        node (f121) {$\fibG{\fibChanB_{3}}$}
        \fibinftrees{c}
      }
      child{
        node (f123) {$\fibG{\fibChanB_{3}}$}
        \fibinftrees{d}
      }
    }
  }
  child{
    node (a1) {$\tyrecv{\color{\fibocolorA}{\fibChan}}$}
  } 
  ;
          \begin{scope}[on background layer]
                                    \path[draw=\fibocolorB,fill=\fibocolorB!5,fill opacity=0.7,densely dotted,rounded corners]
    (f1.north west) -- (f11.north west)
    -- (f11.south west)
    -- (f11.south east)
    -- (b1c1a.north west)
    -- (b1c1a.south west)
    -- (b1c1a.south east)
    -- (f122.south east)
    -- (f122.north east)
    -- (f13.south west)
    -- (f13.south east)
    -- (f13.north east)
    -- (f1.north east)
    -- cycle
    ;
    \path[draw=\fibocolorA,fill=\fibocolorA!5,fill opacity=0.7,densely dotted,rounded corners]
    (root.north east) -- (a1.north east)
    -- (a1.south east)
    -- (f12.north east)
    -- (f12.south east)
    -- (f12.south west)
    -- (f1.south west)
    -- (f1.north west)
    -- (root.north west)
    -- cycle
    ;
                                      \end{scope}
  \end{tikzpicture}
\end{center}
Two fences are highlighted in the diagram: the
$\{{\color{\fibocolorA}{\fibChan}}\}$-fence includes four parallel
types (including the initial $\typevara_0$); while the
$\{{\color{\fibocolorB}{\fibChanB_1}}\}$-fence includes five components:
one instance of $\fibGN(x)$ as well as two of its recursive children
and three instances of the non-recursive component of $\fibGN$.

 \subsection{Symbolic Semantics} \label{sub:sym-sem-tyes}
For {\bf Step 2}, we introduce a symbolic semantics for types, which is parameterised by
a bound on the number of free names that can be used when unfolding a
recursive call, e.g.\ $\tdefcall{\typevara}{\tilde{u}}$, by its
corresponding definition.
The overall purpose of the symbolic semantics is that for any
$\EqTypes$ such that $\tyfenced{\EqTypes}$, the symbolic LTS of
$\EqTypes$ is \emph{finite state}.

\begin{figure}
  \[
  \begin{array}{c}
    \typesymsemrule{def}
    {
      \typesymsemF{T} \subs{\tilde{\tychana}}{\tilde{\tyvar}} \symsem{\labcom} \typesymsemF{T'}
      \ \
      \typevara(\tilde{\tyvar}) = T       \ \
      \tilde{\tychana} \cap \symsemCur \neq \emptyset
    }
    {\tdefcall{\typevara}{\tilde{\tychana}}}{\labcom}{T'}
        \\[1.5em]
    \typesymsemruleC{\resoneless}
    {
      \typesymsem{\symsemCur \uplus \{\tychana\}}{T} 
      \symsem{\labcom} 
      \typesymsem{\symsemCur \uplus \{\tychana\}}{T'} 
      \ \
      \freenames{\labcom}  \neq \{\tychana\} 
      \ \
      \lvert \symsemCur \rvert < k
    }
    {\symsemCur}{(\nu \tychana) T}{\labcom}{\symsemCur}{(\nu \tychana)T'}
    \\[1.5em]
    \typesymsemruleC{\restwoless}
    {
      \typesymsem{\symsemCur \uplus \{\tychana\}}{T} 
      \symsem{\labsync{\tychana}} 
      \typesymsem{\symsemCur \uplus \{\tychana\}}{T'}
      \quad
      \lvert \symsemCur \rvert < k
    }
    {\symsemCur}{(\nu \tychana) T}{\tau}{\symsemCur}{(\nu \tychana)T'}
    \\[1.5em]
    \typesymsemruleC{\resonegeq}
    {
      \typesymsemF{T} \symsem{\labcom} \typesymsemF{T'} 
      \quad
      \freenames{\labcom}  \neq \{\tychana\} 
      \quad
      \lvert \symsemCur \rvert \geq k
    }
    {\symsemCur}{(\nu \tychana) T}{\labcom}{\symsemCur}{(\nu \tychana)T'}
    \\[1.5em]
    \typesymsemruleC{\restwogeq}
    {
      \typesymsemF{T} \symsem{\labsync{\tychana}} \typesymsemF{T'}
      \quad
      \lvert \symsemCur \rvert \geq k
    }
    {\symsemCur}{(\nu \tychana) T}{\tau}{\symsemCur}{(\nu \tychana)T'}
                                      \end{array}
  \]
  \caption{Symbolic Semantics for Types.}\label{fig:type-sym-sem}
\end{figure}

The symbolic semantics for types is given in
Figure~\ref{fig:type-sym-sem}, where we show only the interesting new
rules.
The other rules are essentially those of
Figure~\ref{fig:types-sem}, with the additional parameters
$\symsemidx$ and $\symsemCur$ as expected.
Rule \symrulename{def} replaces its counterpart from
Figure~\ref{fig:types-sem}, while rules \rulename{\resoneless} and\
\symrulename{\resonegeq} replace rule \symrulename{res-1}, and rules
\symrulename{\restwoless} and \symrulename{\restwogeq} replace
\symrulename{res-2}.

In a term $\typesymsemF{T}$, $\symsemCur$ can be seen as a subset of
the free names of $T$.
Whenever a new name is encountered, e.g.\ through $(\nu \tychana)T$,
$\tychana$ is recorded in $\symsemCur$ as long as $\symsemCur$ has
less than $k$ elements.
Rule \symrulename{def} states that a recursive call can only be unfolded
if some of its parameters are in $\symsemCur$.
Note that, in rule \symrulename{def}, we assume that the unfolding of a
type is such that there is no clash with the names in $\symsemCur$.

For $k \geq 0$, we write $\typesymsemF{T} \symsemTr \symsem{\alpha}$
if there exist $T'$ and $T''$ such that $\typesymsemF{T} \symsemTr
\typesymsemF{T' } \symsem{\alpha} \typesymsemF{T''}$.

We consider a fragment of the prime sieve example and show its
behaviour according to the symbolic semantics, with $k=1$.
We have:
\[
\begin{array}{l}
\typesymsem{\{a\}}
{
  (\nu {b})
  (
  \primeG{a}
    \parr
    \primeF{a}{b} \parr \primeR{b})
}
\\
\qquad
\symsemTrI{1} \symsemI{\labsync{a}}{1}
\typesymsem{\{a\}}
{
  (\nu {b})
  (
  \primeG{a}
    \parr
   \tysend{b}; \primeF{a}{b} 
    \parr
    \primeR{b}
  )
}
\end{array}
\]
at this point the process is stuck.
The sub-term $\tysend{b}; \primeF{a}{b}$ awaits to synchronise on
$b$, however the dual action on $b$ is ``hidden'' in the unfolding of
$\primeR{b}$, which cannot be unfolded by rule \symrulename{def} since
$b \notin \{a\}$ and, since $k=1$, $b$ cannot be added to the set of names.
If we set the bound to $k=2$:
\[
\begin{array}{l}
  \typesymsem{\{a, b\}}
  {
    \primeG{a}
        \parr
        \tysend{b}; \primeF{a}{b} 
        \parr
        \primeR{b}
  }
  \\
  \qquad \symsemTrI{2} \symsemI{\labsync{b}}{2}
  \typesymsem{\{a, b\}}
  {
    \primeG{a}
        \parr
        \primeF{a}{b} 
        \parr
        \tynew{c}
    (\primeF{b}{c} \parr \primeR{c})
  }
\end{array}
\]

 \subsection{Liveness and Channel Safety for Types}\label{sub:ty-live-safe}
Following \S~\ref{sub:live-safe-process}, we
define liveness and channel safety properties for types.
The definitions of liveness and channel safety rely on
barbs. \label{sub:barbs-types-app}
The predicate $T\barb{\ov{\tychana}}$ (resp.\
$T\barb{\tychana}$) denotes a type ready to send
(resp.\ receive) over channel $\tychana$.
Barb $T\barb{\clbarb{\tychana}}$ denotes a type ready to close channel
$\tychana$ and barb $T\barb{\cclbarb{\tychana}}$ denotes a closed
channel.
Barb ${T\barb{\labsync{\tychana}}}$ denotes a synchronisation over
channel $\tychana$.
Barb $T\barb{\tilde{o}}$ denotes a type that is waiting to synchronise
over the actions in $\tilde{o}$.
\begin{definition}[Type Barbs]~\label{def:type-barbs}
We define the predicates
$\typrefix\barb{o}$, 
$T\barb{o}$ and $T\barb{\tilde{o}}$ 
with $o,o_i \in \{
\labtyrcv{\tychana}, 
\labtysnd{\tychana}, 
\labsync{\tychana},  
\labclose{\tychana}, 
\labclsnd{\tychana}\}$.
\[
\tyrecv{\tychana}\barb{\tychana}
\qquad
\tysend{\tychana}\barb{\ov{\tychana}}
\qquad
\inferrule
{\typrefix\barb{o}}
{\typrefix ; T \barb{o}}
\qquad
\labclose{\tychana};T\barb{\clbarb{\tychana}}
\qquad 
\tyclosedbuffer{\tychana}\barb{\cclbarb{\tychana}} 
\]
\[
\inferrule
  {T\barb{o}}
  {T\parr T' \barb{o}}
\quad 
\inferrule
{T\barb{o}\ \ \tychana \notin \fn(o)}
{(\nu \tychana)T\barb{o}}
\quad 
\inferrule
{T\barb{o}\ \ T \equiv T'}
{T\barb{o}}
\]
\[
\inferrule
{T\subs{\tilde{\tychana}}{\tilde{x}} \barb{o}\quad \typevara(\tilde{x}) = T}
{{\tdefcall{\typevara}{\tilde{\tychana}}}\barb{o}}
\]
\[
\inferrule
{\forall i \in \{1, \ldots, n \} \qst \typrefix_i \barb{o_i} }
{\tbr{}{ \typrefix_i ; T}{i\in \{1, \ldots, n \}}\barb{ \{ o_1 \ldots o_n \} }}
\quad 
\inferrule
{T\barb{\tychana}\quad T'\barb{\ov{a}}\mbox{ or }T'\barb{\cclbarb{\tychana}}}
{T\parr T'\barb{\labsync{\tychana}}}
\]
\[
\inferrule
{T\barb{\tychana}\quad \typrefix_i\barb{\ov{a}}}
{T\parr \tbr{}{\typrefix_i ; S_i}{i \in I}\barb{\labsync{\tychana}}}
\quad 
\inferrule
{T\barb{\ov{\tychana}}\mbox{ or }T\barb{\cclbarb{\tychana}}\quad \typrefix_i \barb{\tychana}}
{T\parr \tbr{}{\typrefix_i ; S_i}{i \in I} \barb{\labsync{\tychana}}}
\]
\end{definition}

Given $k \in \naturals$, we write ${T}\wbarbsym{o}{k}$ if
$\typesymsem{\symsemCur}{T} \symsemTrI{n}
\typesymsem{\symsemCur}{T'}$ and ${T'}\barb{o}$, with $\symsemCur =
\freenames{T}$, $n = k + \lvert N \rvert$ and $o \in \{\tychana , \ov{\tychana} ,
{\labsync{\tychana}}, \clbarb{\tychana}, \cclbarb{\tychana}\}$.
Observe that the predicate ${T}\wbarbsym{o}{k}$ is defined wrt.\ the
symbolic semantics. 
We write ${T}\wbarbt{o}$ if  ${T}\wbarbsym{o}{\infty}$.

\begin{definition}[Liveness]\label{def:type-live}\rm
  The system $\EqTypes$ satisfies {\em $k$-liveness} if for all $T$ such that
  $\typesymsem{\emptyset}{\typevara_0\ENCan{}}
  \symsemTr \typesymsem{\emptyset}{(\nu \tilde{a}) T}$,
    
  \begin{itemize}
  \item[(a)] If ${T}\barb{\tychana}$ or 
    ${T}\barb{\ov{\tychana}}$
    then 
    ${T}\wbarbsym{\labsync{\tychana}}{k}$.  
      \item[(b)] If ${T}\barb{\tilde{\tychana}}$
    then ${T}\wbarbsym{\labsync{\tychana_i}}{k}$ for some $\tychana_i\in \tilde{\tychana}$.   
  \end{itemize}

   If $\EqTypes$ is $\infty$-live, then we say
 that $\EqTypes$ is \emph{live}.
  \end{definition}

Consider the type equations below, where the reduction of
$\typevara_0\ENCan{}$ leads to terms that are \emph{not}
$k$-live, for any $k \geq 0$.
\[
\begin{array}{lcl}
  \typevara_1(x)  \triangleq   \tynew{b} (   \tdefcall{\typevara_1}{b} \parr   \tyrecv{b}; \tysend{x})
    \quad
  \typevara_0()   
\triangleq 
  \tynew{a}
  (
  \tdefcall{\typevara_1}{a}
  \parr
  \tyrecv{a}
  )
  \end{array}
\]
Intuitively, the system is not live since it is not
possible to find a synchronisation for, e.g.\ the receive action on
$a$, within bounded unfolding.

The definition of channel safety follows the same structure as
that of liveness.

\begin{definition}[Channel Safety]~\label{def:type-safe}\rm 
  The system $\EqTypes$ satisfies {\em $k$-channel safety} if for all $T$ such that
  $\typesymsem{\emptyset}{\typevara_0\ENCan{}}
  \symsemTr \typesymsem{\emptyset}{(\nu \tilde{a}) T}$,
    if 
  $T\barb{\cclbarb{\tychana}}$
  then 
  $\neg(T\wbarbsym{\clbarb{\tychana}}{k})$ and $\neg (T\wbarbsym{\ov{a}}{k})$.
  
  If $\EqTypes$ is $\infty$-safe, then we say
  that $\EqTypes$ is \emph{channel safe}.
\end{definition}

\paragraph{Example}
We illustrate the need of a sufficiently large
bound to detect liveness errors with an example.
Consider a variation of the Prime Sieve example with a non-recursive
filter:
\[
\begin{array}{rcl}
        \primeFN(x,y) & \triangleq & 
  \tyrecv{x};\tysend{y};
  \tyrecv{x};\tysend{y};
  \tyrecv{x};\tysend{y};
  \tyrecv{x};\tysend{y};  \zero
              \end{array}
\]
with $\primeGN(x)$, $\primeRN(x)$, and
$\typevara_0()$ as in \S~\ref{sec:overview}.
The system above is $2$-live, but not $3$-live.
A bound of $2$ is too small to allow the symbolic
semantics to explore all the states of $\primeFN(x,y)$;
while a bound of $3$ enables spawning another filter process,
hence to explore all the states of $\primeFN(x,y)$.

 \subsection{Decidability of the Bounded Verification}\label{sub:proof-types}\fxnote{edited}

We show that $k$-liveness and $k$-channel safety are decidable, noting
that these are defined wrt the symbolic semantics ({\bf Step 3}).
\fxnote{edited}

The \emph{decidability} result mainly follows from fencing,
which ensures finite control over the symbolic semantics.
The key idea is to show that the set of non-equivalent terms
reachable by the symbolic semantics is finite. This is formalised
below.

\begin{restatable}[Finite Control]{lemma}{lemfinitelts}
\label{lem:finite-lts}
  If $\tytyJudgeE{\EqTypes}$, then the set
  \\
  $\left\{ [T]_{\equiv} \mid \typesymsem{\emptyset}{ {\typevara_0\ENCan{}} }
  \symsemTr \typesymsem{\emptyset}{T} \right\}$  
  is finite, for any finite $k$.
  \end{restatable}
The crux of the proof is to show that the number of occurrences of a
type variable $\typevara$ is bounded in the maximal unfolding of a
term up-to a given set of names $\symsemCur$. That is, an unfolding
which unfolds a term $\tdefcall{\typevara}{\tilde{\tychana}}$ only if
$\tilde{\tychana} \cap \symsemCur \neq \varnothing$.

\begin{restatable}[Decidability]{theorem}{thmdecidability}
  \label{thm:decidability-ty}
  For all $\EqTypes$ s.t.\  $\tytyJudgeE{\EqTypes}$, it is
  decidable whether or not $\EqTypes$ is $k$-live (resp. channel safe), for
  any $k \geq 0$.
\end{restatable}
Theorem~\ref{thm:decidability-ty} follows from the fact that checking
$k$-liveness (resp.\ channel safety) is decidable over any finite LTS
(finiteness is guaranteed by Lemma~\ref{lem:finite-lts}).

\begin{remark}\rm
  It is not always possible to compute a finite $k$ such that
  $k$-liveness (resp.\ $k$-channel safety) implies general liveness
  (resp.\ channel safety) of fenced types.
    However, if the types are \emph{finite control} (i.e., parallel
  composition does not appear under recursion), then liveness and
  channel safety are indeed decidable, see e.g.~\cite{finitecontrolccs}.
\end{remark}

 \section{Properties of \fgo}
\label{sec:properties}
We now make precise the properties our behavioural type analysis
ensures on \fgo programs. We show that if a program $\prog$ is typed
by a safe type, then $\prog$ is safe according to
Definition~\ref{def:psafe}. For liveness, we identify the classes of
programs for which liveness of types implies program liveness. We note
that type liveness and safety refers to $\infty$-liveness and
$\infty$-safety, respectively. Moreover, we recall that our bounded
analysis on types implies its $\infty$-counterpart only in the finite
control fragment.

\subsection{Type and Channel Safety in \fgo}\label{subsec:typesafety}
The typing system of \S~\ref{sect:core_typ} ensures that channel
payloads always have the expected type. This property is made precise
by a standard subject reduction result, stating that the semantics of
types simulates the semantics of processes.

\begin{restatable}[Subject Reduction]{theorem}{thmsubreduction}
  \label{thrm:subred}
Let $\G \vdash_B P \ts T$ and $P \tra{} P'$. Then 
there
exists $T'$ such that 
$\G \vdash_{B'} P' \ts T'$ with $T\tra{} T'$. 
\end{restatable}

We prove that type safety implies program safety, using 
a correspondence between barbs.
Hereafter we write ${P}\wbarb{o}$ including the case $o=\tilde{a}$.

\begin{lemma}\label{lem:barb}
Suppose $\G \vdash_B P \ts T$. 
If ${P}\wbarb{o}$ then 
${T} \wbarbt{o}$. 
\end{lemma}

\begin{restatable}[Process Channel Safety]{theorem}{thmchansafe}
  \label{thm:chansafe}
  Suppose $\G \vdash \prog \ts \EqTypes$ and  
  $\EqTypes$ is safe. Then $\prog$ is safe. 
\end{restatable}

\subsection{Liveness of Limited Programs}
\label{sec:liveness:challenges}
The development in \S~\ref{sec:types} performs an analysis on our
abstract representation of processes (i.e. the \emph{types}),
verifying liveness (Definition~\ref{def:type-live}) for fenced types.
Our goal is to ensure liveness of a general class of
\emph{typable programs}. We divide programs into three classes to
discuss the issue of liveness.

The first class is a set of programs which have a path to terminate.
In this class, a program that is typable with a live type can always satisfy
liveness.

\begin{definition}[May Converging Program] \label{def:termination}\rm
Let $\G \vdash \prog \ts \EqTypes$. 
We write $\prog \in \convm$ if 
for all $X_0\ENCan{} \tra{}^\ast P$, $P \tra{}^\ast \zero$.  
\end{definition}

\begin{restatable}{proposition}{proptermination}
  \label{pro:termination}
  Assume $\G \vdash \prog \ts \EqTypes$ and 
  $\EqTypes$ is live. 
      {\em (1)} 
  Suppose  
  there exists $P$ 
  such that $X_0\ENCan{}\tra{}^\ast P \not\tra{}$.
  Then $P\equiv \zero$; and {\em (2)} If $\prog\in \convm$, then 
  $\prog$ is live. 
      \end{restatable}
\noindent 
We note that the above statement does not restrict the programs to be
finite. 
A program with infinite reduction sequences 
can satisfy liveness by Proposition~\ref{pro:termination}.
For instance:
\[
\pin{\{ D_1,D_2 \}}{\newch{b};\newch{c};
(X_1\langle b,c\rangle \parr
  X_2\langle b,c \rangle)}
\]
\[
\begin{array}{rcl}
  \hspace{-2mm}\mbox{with}\quad
D_1 & = 
&  X_1(b,c) \triangleq  
\sel{\psend{c}{v};X_1\langle b,c
  \rangle \, ,\, \psend{b}{w};\zero}{}\\
D_2 & = 
& X_2(b,c) \triangleq \sel{\precv{c}{x};X_2\langle b,c
  \rangle \, , \, \precv{b}{y};\zero}{}
\end{array}
\]
is live, as is its corresponding type. 

The second class is a set of programs which do not contain
\emph{infinitely occurring conditional branches} (we discuss at length
the issues raised by the interplay of conditional branching and
recursion in \S~\ref{subsec:liveness}).  If such a process is assigned
a live type, then it is itself live.  For example, any program which
does not include conditionals or one with conditionals containing only
finite processes in both branches belong to this class.  Consider a
program obtained from the one above by replacing $\zero$ in $D_1$ and
$D_2$ by $X_1$ and $X_2$, respectively.  Despite this program
executing forever, both program and type liveness hold.  This
class of programs is made precise in Proposition~\ref{pro:liveness},
along with the issue of infinitely occurring conditionals, which are
explained below.

\subsection{Liveness of Infinitely Occurring Conditionals}
\label{subsec:liveness}
As the third class, we consider infinitely running
programs that contain 
recursive variables in conditional branches.  
The behaviours of conditionals in a program rely on data to decide 
which branch is taken. On the other hand,  
at the type level,    
this information is abstracted as an internal choice. 
This causes a mismatch between 
program and type behaviours. 

Revisiting the prime sieve example of \S~\ref{sect:ex_proc_sieve}, 
consider the definition of the filter process:
\[
\begin{array}{rcl}
F(n,i,o) & \triangleq & \precv{i}{x};\pite{(x\% n \neq 0)}{\psend{o}{x};F\langle
               n,i,o \rangle\\ & &\hspace{2.6cm}}{\ F\langle n , i , o\rangle}\\
\end{array}
\]
whose type is given as:
$\mathbf{t}_F(i,o) = \trecv{i}{};\oplus\{
\tsend{o}{};\mathbf{t}_F\langle i,o \rangle,\, \mathbf{t}_F\langle i,o
\rangle \}$.

Our analysis on types does indeed determine
the types of the prime sieve as live,  even in the absence of
terminating reduction sequences. 
In 
$\mathbf{t}_F(i,o)$, we have an internal choice between a branch
that recurses back to $\mathbf{t}_F$ and another that outputs along
$o$ and recurses back to $\mathbf{t}_F$. Thus, if we compose a call
$\mathbf{t}_F$ with a type that denotes an infinite
sequence of inputs along $o$, we deem such a composition as live since
all the inputs \emph{can} eventually be synchronised with an output
from $\mathbf{t}_F$, given that the semantics of internal choice state
that we may indeed move to either branch. 

However, the type $\EqTypes$ of the prime sieve program 
is an abstract approximation of the actual
prime sieve implementation, where the test $x\% n \neq 0$ is not
obviously guaranteed to ever succeed given that it depends on received
data (which is sent by either the generator process $G$ or a previous
filter process). Thus, the interplay of conditional branching and
infinite recursion may in general cause a disconnect between the semantics of
the types and those of the concrete processes.
For instance, if the test $x\% n \neq 0$ is replaced by $\false$ in
the prime sieve example, its type is live while the program is not.
In the remainder of this section, we 
make precise the conditions under which the semantics of infinite processes
and types simulate one another, thus implying
liveness (even in the presence of infinite branching).

To achieve our liveness results, we proceed as follows:
\begin{description}
\item[Step 1.] Define the notion of {\em infinite conditional}
  ($\ic$), identifying a class of programs where conditional branches
  are executed infinitely often.

\item[Step 2.] Fill the gap between internal choices of types 
and conditionals by defining a $\ast$-conditional 
($\pite{\ast}{P}{Q}$)
which non-deterministically reduces to either $P$ or $Q$ (as the internal choice
$\oplus \{T,S\}$), allowing us to identify the subclass of $\ic$,
dubbed alternating conditionals ($\ac$), where
programs simulate their non-deterministic conditional counterparts.

\item[Step 3.] Prove that liveness of types  implies
  liveness of programs in $\ac$.
\end{description}

\paragraph{Alternating and Non-deterministic Conditionals}
For {\bf Step 1},  we begin by defining a notion of \emph{infinite
 conditional} (Definition~\ref{def:ic}) in programs. Intuitively, we
identify programs that reduce forever and where conditional
branches appearing under recursion have their branches taken
infinitely often.

\begin{definition}[Marked Programs]\label{def:mark}
  Given a program $\prog$ we define its \emph{marking}, written
  $\m{mark}(\prog)$, as the program obtained by deterministically
  labelling every occurrence of a conditional of the form
  $\pite{e}{Q}{R}$ in $\prog$, as $\mpite{e}{Q}{R}{n}$, such that $n$
  is distinct natural number for all conditionals in $\prog$.
\end{definition}

\begin{definition}[Marked Reduction Semantics]\label{def:mred}
\rm
We define a marked reduction semantics, written $P
\markred{l} Q$, stating that program $P$ reduces to $Q$ in a single
step, performing action $l$. The grammar of action labels is defined as:
\[
l \coloneqq \emp \mid \iflab{n}{\llab} \mid \iflab{n}{\rlab}
\]
\noindent where $\emp$ denotes an unmarked action,
$\iflab{n}{\llab}$ denotes a conditional branch marked
with the natural number $n$ in which
the $\m{then}$ branch is chosen, and $\iflab{n}{\rlab}$
denotes a conditional branch in which the $\m{else}$ branch is
chosen. We write $P\tra{} Q$ for $P\tra{\emp} Q$.  The marked
reduction semantics replace rules \rulename{ift} and \rulename{iff}
with:
\[
\begin{array}{c}
\small
\rulename{iftm}
    \inferrule
    {e\conv \true}
    {{\mpite{e}{P}{Q}{n}}\tra{\iflab{n}{\llab}}
    {P}}
\quad
\rulename{iffm}
    \inferrule
    {e\conv \false}
    {{\mpite{e}{P}{Q}{n}}\tra{\iflab{n}{\rlab}}
    {Q}}
\end{array}
\]
\end{definition}

\begin{definition}[Trace]\rm 
We define an execution trace $\trace$ of a process $P$ as the potentially
infinite sequence of action labels $\vec{l}$ such that $P\markred{l_1}
P_1 \markred{l_2} \dots$, with $\vec{l} = \{l_1 \, l_2 \dots \}$.
We write $\traceset_P$ for the set of all possible traces of a process $P$.
\end{definition}

A trace of the marked reduction semantics identifies exactly which
branches were selected during the potentially infinite execution of a
program.

We now define infinitely recurring conditionals.
We use a reduction context $\tycontxt_r$ given by:
\[
\begin{array}{lcl}
\tycontxt_r & := & [\,] \,\mid\, (P\parr \tycontxt_r)  \,\mid\,  
   (\tycontxt_r \parr P) \,\mid\, (\nub a)\tycontxt_r
\end{array}
\]
We write $\tycontxt_r[P]$ for the process obtained by replacing $P$
for the hole $[\,]$ in $\tycontxt_r$.

\begin{definition}[Infinite Conditional]~\label{def:ic}
\rm We say that $P$ has infinite conditional branches, written $P\in\ic$,
iff $\m{mark}(P) \tra{}^*
\tycontxt_r[\mpite{e}{Q_1}{Q_2}{n}] = R$, for some $n$,
and $R$ has an infinite trace where
$\iflab{n}{\llab}$ or $\iflab{n}{\rlab}$ appears infinitely often.
We
say that such an $n$ is an \emph{infinite conditional mark} and write
$\infcond(P)$ for the set of all such marks.
\end{definition}

The following statement implies that even programs which contain only
infinite executions can be live if none of its conditionals appear in
traces infinitely often (i.e. our second class of programs).

\begin{restatable}[Liveness for Finite Branching]{proposition}{propliveness}
  \label{pro:liveness}
  \label{pro:finite:liveness} 
  Suppose $\G \vdash \prog \ts \EqTypes$ and $\EqTypes$ is live 
  and $\prog\not\in \ic$.
  Then $\prog$ is live.
    \end{restatable}

\noindent The main purpose of Definition~\ref{def:ac} is to identify
infinitely running processes where the behaviour of conditional
branching approximates that of non-deterministic internal choice
(i.e. the type-level semantics of internal choice). To make this
relationship precise, we define a mapping from \fgo programs to programs
where conditional branching is replaced by a form of
non-deterministic branching. This step corresponds to {\bf Step 2}.

\begin{definition}\label{def:starmapping}\rm
The mapping $(P )^\ast$ replaces all occurrences of 
$\mpite{e}{Q}{R}{n}$, such that $n
\in\infcond(P)$, with $\pite{\ast}{Q}{R}$.
The reduction semantics of $\pite{\ast}{Q}{R}$ is defined as follows: 
\[
\begin{array}{ll}
\hspace{-2mm}
\rulename{ift$\ast$}\ \ 
\pite{\ast}{P}{Q}\tra{}P\quad 
\rulename{iff$\ast$}\ \ 
\pite{\ast}{P}{Q}\tra{}Q 
\end{array}
\]
\end{definition}

\begin{definition}[Alternating Conditionals]~\label{def:ac}\rm We say
  that $P$ has {\em alternating conditional branches}, written
  $P\in\ac$, iff $P\in \ic$ and 
  if $P \tra{}^* (\nub \tilde{c}) Q$ then
  $\MAPAST{Q}\wbarb{o}$ implies $Q \wbarb{o}$.
  \end{definition}

 Recall that $o$ ranges over
any barbs, including $\tilde{a}$.
Moreover, observe that the mapping $\MAPAST{P}$ only affects
conditionals that are executed infinitely often (i.e. those whose
behaviour may fail to be captured by the type-level analysis).  We do
not require conditionals that are not in $\infcond(P)$ to necessarily
match the barbs of their non-deterministic counterpart, since their
behaviour is already over-approximated by the corresponding types.

\begin{proposition}[$\ast$-properties]\label{prop:starprop}
\label{lem:ast-barb}
Suppose $\G \vdash_B P \ts T$. Then {\em (1)} 
if $\MAPAST{P}\in \ic$ then $\MAPAST{P}\in \ac$; 
{\em (2)}  If $P \wbarb{o}$, then  
$\MAPAST{P} \wbarb{o}$; {\em (3)}     
if $\MAPAST{P} \wbarb{o}$ then $T\wbarbt{o}$.
\end{proposition}

\paragraph{Liveness for Infinite Conditionals}
We now have defined the conditions under which programs simulate the
behaviour of their types. More precisely, when a program $\prog$ is
well-typed with some live type $T$ and $\prog\in \ac$ holds, then
$\prog$ must itself be live.

\begin{restatable}[Liveness]{theorem}{thmliveness}
  \label{thm:liveness}
  \label{theorem:liveness} 
  Suppose $\G \vdash \prog \ts \EqTypes$ and $\EqTypes$ is live 
  and $\prog\in \ac$.
  Then $\prog$ is live.
    \end{restatable}
To summarise, we identified three significant classes of programs for
which type liveness implies liveness: those with at least
one terminating path (Definition~\ref{def:termination} and
Proposition~\ref{pro:termination}) such as Fibonacci, cf.\ \S~\ref{sec:sync-fib}; 
those for which their infinite traces
do not contain infinite occurrences of a given conditional
(Proposition~\ref{pro:liveness})  such as
Dining Philosophers, cf.\ \S~\ref{sec:eval}; and, 
those with infinite traces containing infinite occurrences of
conditional branches (Definition~\ref{def:ac} and
Theorem~\ref{thm:liveness}) such as Prime Sieve, cf.\
\S~\ref{sect:ex_proc_sieve}.

While a reasonable percentage of real-world programs are in the first
two classes, our empirical observations show that a substantial amount
of infinitely running programs (with infinitely occurring
conditionals) that are not in $\ac$ have redundant or erroneous
conditionals.

 \section{Bounded Asynchrony in \fgo}
\label{sec:asynchrony}
Our framework extends with relative ease to the asynchronous
communication variant of the Go language. As mentioned in
\S~\ref{sect:lang}, communication channels in Go are implemented as
\emph{bounded} FIFO queues, where by default the buffer bound is $0$ --
synchronous communication. For bounds greater than $0$, communication
is then potentially asynchronous -- sends do not block if the buffer
is not full and inputs do not block if the buffer is not empty.

Asynchrony significantly affects a program's liveness. Consider the
following example:
\[
\begin{array}{lcl}
P(x,y) & \triangleq & \psend{x}{1};\precv{y}{z} \parr \psend{y}{2};\precv{x}{z}
\end{array}
\]
A program that instantiates $P(x,y)$ with synchronous communication
channels will necessarily not be live since the output and input actions
in $P$ are mismatched. However, with asynchronous channels, the output
actions become non-blocking and the program is indeed \emph{live} --
the output on $x$ on the left-hand side can fire asynchronously, exposing the
input on $y$ which may then fire. Similarly for the output on $y$ and
input on $x$ on the right-hand side.

\paragraph{Processes and Typing} To account for the buffer bounds in the syntax of \fgo we
add a bound $n$ to channel creation, $\newch{y{:}\sigma,n};P$. This
number must be equal or greater to zero and must be a literal. We also
carry this information in runtime 
buffers: $\buff{\tilde{\vval}}{c}{\sigma,n}$ and
$\closedbuff{\tilde{\vval}}{c}{\sigma,n}$ 
(also replacing $\buff{\emptyset}{c}{\sigma}$ and 
$\closedbuff{\tilde{\vval}}{c}{\sigma}$ 
by $\buff{\emptyset}{c}{\sigma,0}$ 
and $\closedbuff{\tilde{\vval}}{c}{\sigma,n}$ for synchronous
channels).  
We add
the reduction rules for asynchronous communication:
\[
\begin{array}{c}
\rulename{out}
 \inferrule {\buflen{\tilde{\vval}} < n\quad e \downarrow v}
 {{\psend{c}{e};P \parr \buff{\tilde{\vval}}{c}{\sigma,n}}\tra{}
 {P \parr \buff{\vval \cdot \tilde{\vval}}{c}{\sigma,n}}}
 \\[1.5em]
\rulename{ina}\quad 
 {\precv{c}{y};P \parr \buff{\tilde{\vval} \cdot \vval }{c}{\sigma, n}}\tra{}
 {P\subs{\vval}{y} \parr \buff{\tilde{\vval} }{c}{\sigma, n} }
\end{array}
\]
In all other rules that use buffers we add the buffer bound
straightforwardly.
The type system is fundamentally unchanged, now accounting for
buffer bounds:
\[
\begin{array}{c}
\trulename{new}
\inferrule
{\G , y{:}\m{ch}(\sigma,n) \vdash_s P \ts T \quad c\not\in \domain{\G}\cup s
  \cup \fn(T)}{\G \vdash_{s} {\newch{y{:}\sigma,n}};P \ts
  \astynew{c}{n}T\subs{c}{y} }\\[1em]
\trulename{buff}
\inferrule{\lvert \tilde{\vval} \rvert = k}
{\G ,a{:}\m{ch}(\sigma,n)
  \vdash_{\{a\}}\buff{\tilde{\vval}}{a}{\sigma, n}  \ts \astyOpenbuf{\tychana}{k}{n}}
\end{array}
\]
In contrast with the types and rules in Figure~\ref{fig:typing}, 
${\tynew{\tychana}T}$ and $\tyopenbuffer{\tychana}$ are replaced by
$\astynew{\tychana}{n}T$ and ${\astyOpenbuf{\tychana}{k}{n}}$,
respectively, where $k$ stands for the number of elements in the
buffer.

Liveness and safety of types are defined as in
\S~\ref{sub:ty-live-safe}, with two extra rules for the definition of
type barbs, pertaining to buffers. In particular we need barbs for writing to
a non-full buffer ($P\barb{\astyBufSend{\tychana}}$) and reading from
a non-empty buffer ($P\barb{\astyBufRcv{\tychana}}$), combined with
  the following additional rules:
\[
\begin{array}{c}
\inferrule[]
{\buflen{\tilde{\vval}} < n}
{\buff{\tilde{\vval}} {a}{\sigma,n} \barb{\astyBufSend{\tychana}}}
  \quad
\inferrule[]
{\buflen{\tilde{\vval}} \geq 1}
{\buff{\tilde{\vval}}{a}{\sigma,n} \barb{\astyBufRcv{\tychana}}}
\\[5mm]
\inferrule
{P\barb{\ov{a}}\quad Q\barb{\astyBufSend{\tychana}}}
{P\parr Q\barb{[a]}}
\quad
\inferrule
{P\barb{\astyBufRcv{\tychana}} \quad \pi_i\barb{a}}
{P\parr \sel{\pi_i ; Q_i}{i \in I} \barb{[a]}}
\end{array}
\]
The barbs for asynchronous types, $T\barb{o}$, 
are given below:
\[
\inferrule
{k < n}
{{\astyOpenbuf{\tychana}{k}{n}} \barb{\astyBufSend{\tychana}}}
\quad
\inferrule
{k \leq 1}
{{\astyOpenbuf{\tychana}{k}{n}} \barb{\astyBufRcv{\tychana}}}
\quad
\inferrule
{
  T\barb{\ov{\tychana}}
  \quad
  T'\barb{\astyBufSend{\tychana}}}
{T\parr T'\barb{\labsync{\tychana}}}
\quad
\inferrule
{
  T\barb{{\astyBufRcv{\tychana}}}
      \quad \typrefix_i \barb{\tychana}
}
{T\parr \tbr{}{\typrefix_i ; S_i}{i \in I} \barb{\labsync{\tychana}}}
\]

\paragraph{Verification of Types} The changes to the semantics of types are straightforward. It is based
on the LTS of \S~\ref{sub:sync-ty-sem}, where rule \ltsrulename{new}
and \ltsrulename{buf} are replaced by their counterparts below, and
four additional rules \ltsrulename{in-b}, \ltsrulename{out-b},
\ltsrulename{push} and \ltsrulename{pop}.
\[
\begin{array}{c}
    \ltsrulename{new}\ \  
  {\astynew{\tychana}{n}T}\tra{\tau}
  {(\nu \tychana)(T \parr {\astyOpenbuf{\tychana}{0}{n})}}
  \quad
  \ltsrulename{buf}\ \
  {\astyOpenbuf{\tychana}{k}{n}}\tra{\labclosedual{\tychana}}
  {\tyclosedbuffer{\tychana}}
  \end{array}
\]
\[
\begin{array}{c}
    \ltsrulename{in-b}
  \inferrule
  {k < n}
  {{\astyOpenbuf{\tychana}{k}{n}}\tra{\astyBufSend{\tychana}}{\astyOpenbuf{\tychana}{k+1}{n}}}
  \quad
  \ltsrulename{out-b}
  \inferrule
  {k \geq 1}
  {{\astyOpenbuf{\tychana}{k}{n}}\tra{\astyBufRcv{\tychana}}{\astyOpenbuf{\tychana}{k-1}{n}}}
\end{array}
\]
\[
\begin{array}{c}
    \ltsrulename{push}
  \inferrule
  {\typesemF{T} \semty{\labtysnd{\tychana}} \typesemF{T'} 
    \quad
    \typesemF{S} \semty{\astyBufSend{\tychana}} \typesemF{S'}
  }
  {{T \parr S}\tra{\labsync{\tychana}}
    {T'\parr S'}}
  \quad
  \ltsrulename{pop}
  \inferrule
  {\typesemFone{T} \semty{\astyBufRcv{\tychana}} \typesemFone{T'} 
    \quad
    \typesemFtwo{S} \semty{\labtyrcv{\tychana}} \typesemFtwo{S'}
  }
  {{T \parr S}\tra{\labsync{\tychana}}
    {T'\parr S'}}
\end{array}
\]

Observe that since types abstract away from values and channels are
attributed a unique payload type, the semantics does not model message
ordering.

With all the technical machinery in place for the \emph{bounded}
asynchronous setting, we replicate our main results. The proofs are
essentially identical to those in the synchronous setting.
  Indeed, asynchrony affects our analysis only in the size of the models
to be checked (larger buffer sizes give larger LTSs). The
symbolic semantics executes the types up-to a limited number of
channels, which is orthogonal to the number of message a buffer can
store, cf.\ Figure~\ref{fig:type-sym-sem}.

\begin{theorem}[Decidability -- Asynchrony]\label{thm:decidability-aty}
  For all $\EqTypes$ s.t.\ $\tytyJudgeE{\EqTypes}$, it is decidable
  whether or not $\EqTypes$ is $k$-live (resp. channel safe), for any
  $k \geq 0$.
\end{theorem}

\begin{theorem}[Process Channel Safety and Liveness -- Asynchrony]
\label{thm:a:safety:liveness} 
Suppose $\G \vdash \prog \ts \EqTypes$. 
\begin{enumerate}
\item If $\EqTypes$ is channel safe, then $\prog$ is channel safe.
\item If $\EqTypes$ is live and either 
$\prog\in \convm$, 
$\prog\not\in \ic$ or 
$\prog\in \ac$, then $\prog$ is live.
\end{enumerate}
\end{theorem}
\noindent With the revised semantics, the program
\[
\begin{array}{c}
\pin{\{P(x,y)\}}{\newch{x{:}\m{int},1};\newch{y{:}\m{int},1};P\langle x , y\rangle}
\end{array}
\]
is correctly deemed as live, with the typing given by:
\[
\{\typevara_{P}(x,y) =
(\tsend{x}{};\trecv{y}{} \parr \tsend{y}{};\trecv{x}{}) \} \,
\m{in}\,(\m{new}^1 x)(\m{new}^1 y)\typevara_{P}\langle x , y \rangle
\]

 \section{Implementation}\label{sec:impl}

We have implemented our static analysis as a verification tool-chain consisting
of two parts: First, we analyse Go source code and infer behavioural types
(\S~\ref{sect:core_typ}) based on a program's usage of concurrency
primitives.
The types are passed to a tool that implements the verification outlined in
\S~\ref{sec:types}, checking bounded liveness and channel safety of the types.
An outline
of our verification tool chain is shown in
\figurename~\ref{fig:workflow}.

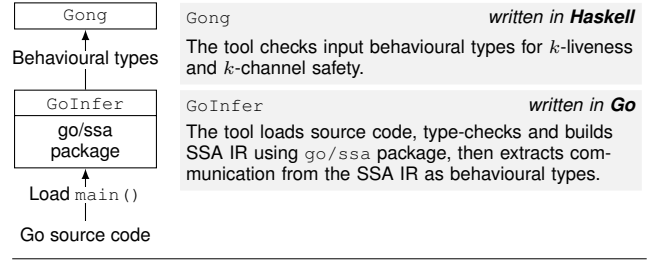
\begin{figure}[th!]
\begin{tikzpicture}[
    tool/.style={draw,text width=7em,text centered},
    arrowlabel/.style={inner sep=0,fill=white},
    comment/.style={text width=\linewidth-10em,anchor=north west,
    fill=gray!10     }
  ]
  \scriptsize\sffamily
  \node (gong) [tool] {\checktool};
  \node (dingo) [rectangle split,rectangle split parts=2,below=0.8 of gong,tool]
        {\infertool\nodepart{two}\scriptsize{\sf go/ssa} package};
  \node (source) [below=0.7 of dingo]
        {Go source code};
  \path [draw,->,>=latex] (source) -- node [arrowlabel] {Load \texttt{main()}} (dingo);
  \path [draw,->,>=latex] (dingo)  -- node [arrowlabel] {Behavioural types} (gong);
  \node [rectangle split,rectangle split parts=2,right=0.3 of gong.north east,comment]
        {\checktool\hfill\textit{written in \textbf{Haskell}}\nodepart{two}
         The tool checks input behavioural types for $k$-liveness
         and $k$-channel safety.};
  \node [rectangle split,rectangle split parts=2,right=0.3 of dingo.north east,comment]
        {\infertool\hfill\textit{written in \textbf{Go}}\nodepart{two}
         The tool loads source code, type-checks and builds SSA IR using
         \texttt{go/ssa} package, then extracts communication from the SSA IR as
         behavioural types.};
\end{tikzpicture}
  \caption{Workflow of our verification tool chain.\label{fig:workflow}}
\end{figure}

\paragraph{Type Inference}  Our type inference tool \infertool is written
in Go, using the
\texttt{go/ssa}\footnote{\url{http://golang.org/x/tools/go/ssa}} package from Go
project's extra tools. The
package builds Go source code in Static Single Assignment (SSA)
representation, and provides an API to access the resulting SSA IR
programmatically. Starting from the program entry point, i.e.\ the
\lstgo{main()} function in the \lstgo{main} package, we transform the
SSA IR into a system of type equations $\EqTypes$ by converting each
SSA block into an individual type equation. The analysis and
conversions are context-sensitive, for example, channels created in
different instances of a function are different, and loops are
unrolled if it is possible to determine the bounds statically.
We note that our analysis is agnostic wrt. aliasing since we do not rely on
linearity of channels.
In addition to inference, our tool can also check for trivial
conditionals that do not belong to any of the three classes of programs
defined in \S~\ref{sec:properties}.

\paragraph{Verification} Our proof-of-concept verification tool
\checktool, written in Haskell, inputs a system of type equations
$\EqTypes$ representing a Go program's concurrent behaviour and
performs liveness and channel safety checks on the bounded symbolic
semantics. Our representation of $\EqTypes$ makes use of the {\tt
  unbound} package~\cite{unbound} to deal with the binding structure
of types.
First, it checks if $\EqTypes$ is fenced. If so, we generate all
$\symsemI{}{k}$-reachable terms, where $k$ is
heuristically-computed. Finally, each of these terms are checked for
$k$-liveness and $k$-channel safety by identifying their barbs and
successors.

\subsection{Evaluation}\label{sec:eval}
We tested our tool-chain on the examples from the paper, from works on
static deadlock detection in Go~\cite{NY2016},
on open-source Go programs from developer
guides~\cite{web:concur-patterns,web:pipelines-md5}
and GitHub~\cite{gh:concsys}, on classic concurrency
examples~\cite{book:concurrency,ewd}, and on concurrent
programs translated to idiomatic Go from from \cite{LTY2015}.

\tablename~\ref{tbl:results} summarises the experimental results.
The column ``Go programs'' shows the names of the programs.
In the columns ``Number of channels'', the number of channels given
for programs with bounded loops is precise since bounded loops are
unrolled and we can statically count the number of channels;
for programs with recursion, we count the channels that appear in the
source code.  The columns ``Time'' show the inference and verification
times in seconds.  We include a comparison with the tool
from~\cite{NY2016} to demonstrate the extra expressiveness of our
approach. If a program is ``Static'', it has no dynamic spawning of
goroutines (a requirement for the usage of the tool of~\cite{NY2016}).

{\tt forselect} is a pattern described in~\cite{web:advconc} where
an infinite {\tt for} loop and a {\tt select} statement with two cases are combined to
repeatedly receive (or send). In our example we spawn two goroutines, where
each goroutine has a for-select loop with compatible channel communication.
In the for-select loop, one of the select cases receives (or sends) a message
then continues to the next iteration of the infinite loop; the other case
breaks out of the loop upon sending (or receiving) a message from the other
goroutine so that both goroutines exit the loop together. The exit condition
is non-deterministic (because of select), but the program is both live
and safe.
{\tt cond-recur} is similar to {\tt forselect}, where one of the two goroutines
contains a for-select loop, but the other has an ordinary for-loop so that the
exit condition of the for-loop is deterministic.

\begin{table}[t]
  {\centering\scriptsize
  \setlength{\tabcolsep}{2.2pt}
  \nocaptionrule\caption{Go programs verified by our tool chain.\label{tbl:results}}
  \begin{tabular}{l|ccccc|c|cc}
  \toprule
    &
    &
    \multicolumn{2}{c}{\# chans}
    &
    &
    & Analysis
    &
    \multicolumn{2}{|c}{\cite{NY2016}}
    \\
    Examples
    & LoC       & unbuf. & buf.
    & Live
    & Safe
    & Time (ms)
    & Static & Safe
    \\
  \midrule
    {\tt sieve}~$\dagger$
    & 19
    & 2
    & 0
    & $\checkmark$
    & $\checkmark$
    & 209.55
    & $\times$
    &
    \\
    {\tt fib}~$\dagger$    & 23
    & 2
    & 0
    & $\checkmark$
    & $\checkmark$
    & 14638.4
    & $\times$
    &
    \\
    {\tt fib-async}~$\dagger$
    & 23
    & 1
    & 1
    & $\checkmark$
    & $\checkmark$
    & 32173.8
    & $\times$
    &
    \\
    {\tt fact}~$\dagger$
    & 19
    & 2
    & 0
    & $\checkmark$
    & $\checkmark$
    & 206.63
    & $\times$
    &
    \\
    {\tt dinephil}~\cite{ewd,book:concurrency}
    & 56
    & 3
    & 0
    & $\checkmark$
    & $\checkmark$
    & 646921.76
    & $\times$
    &
    \\
    {\tt jobsched}
    & 41
    & 0
    & 1
    & $\checkmark$
    & $\checkmark$
    & 48.12
    & $\times$
    &
    \\
    {\tt concsys}~\cite{gh:concsys}
    & 112
    & 2
    & 0
    & $\times$
    & $\checkmark$
    & 323.75
    & $\times$
    &
    \\
    {\tt fanin}~\cite{web:pipelines-md5,NY2016}
    & 36
    & 3
    & 0
    & $\checkmark$
    & $\checkmark$
    & 89.14
    & $\checkmark$
    & $\checkmark$
    \\
    {\tt fanin-alt}~\cite{NY2016}
    & 37
    & 3
    & 0
    & $\times^1$
    & $\checkmark$
    & 209.02
    & $\checkmark$
    & $\checkmark$
    \\
    {\tt mismatch}~\cite{NY2016}
    & 26
    & 2
    & 0
    & $\times$
    & $\checkmark$
    & 26.59
    & $\checkmark$
    & $\times$
    \\
    {\tt fixed}~\cite{NY2016}
    & 25
    & 2
    & 0
    & $\checkmark$
    & $\checkmark$
    & 24.58
    & $\checkmark$
    & $\checkmark$
    \\
    {\tt alt-bit}~\cite{Milner:1989:CC}    & 74
    & 0
    & 2
    & $\checkmark$
    & $\checkmark$
    & 405.78
    & $\checkmark$
    & $\checkmark$
    \\
    {\tt forselect}
    & 40
    & 3
    & 0
    & $\checkmark$
    & $\checkmark$
    & 31.01
    & $\checkmark$
    & $\checkmark$
    \\
    {\tt cond-recur}
    & 32
    & 2
    & 0
    & $\checkmark$
    & $\checkmark$
    & 34.08
    & $\checkmark$
    & $\checkmark$
    \\
    \bottomrule
  \end{tabular}
  {\flushleft
  \noindent
  \noindent
  ${}^{1}$: testing for channel close state is not supported in this version\\
  $\dagger$: examples that are \textit{not} finite control\\
  }
  {\scriptsize
The benchmarks were compiled with ghc 7.10.3 and go1.6.2 executed on Intel Core
i5 @ 3.20GHz with 8GB RAM.
  }}
\vspace{-5mm}
\end{table}

 \section{Related Work and Conclusion}
\label{sec:related}
\paragraph{Static Deadlock Detection in Synchronous Go}
There are two recent works on static deadlock detection for synchronous
Go~\cite{NY2016,Stadmuller16}. The work~\cite{NY2016} extracts {\em
  Communicating Finite State Machines}~\cite{Brand83} whose
representation corresponds to {\em session
  types}~\cite{THK,honda.vasconcelos.kubo:language-primitives} from Go
source code, and synthesises from them a global choreography using the
tool developed in~\cite{LTY2015}. If the choreography is well-formed,
a program does not have a (partial) deadlock.  This approach is
seriously limited due to the lack of expressiveness of (multiparty)
session types~\cite{CHY08} and its synthesis theory. The approach
expects all goroutines to be spawned before any communication happens
at runtime.  This is due to the fact that the synthesis technique
requires all session participants to be present from the start of the
global interaction, meaning that their work cannot handle most
programs with dynamic patterns, such as spawning new threads after
communication started. The analysis is also limited to unbuffered
channels and does not support asynchrony.
For instance, our prime sieve example cannot be verified by their
tool, and is in fact used to clarify the limitations of their
approach. Moreover, the work is limited to the tool implementation, no
theoretical property nor formalisation is studied in \cite{NY2016}.

The work of~\cite{Stadmuller16} uses the notion of {\em forkable
  behaviours} (i.e. a regular expression extended with a fork
construct) to capture spawning behaviours of synchronous Go programs,
developing a tool based on this approach to directly analyse Go
programs.  
Their technique is sound but has some significant
theoretical and practical limitations: (1) their analysis does not
support asynchrony (buffered channels), closing of channels or usages
of the select construct with non-trivial case bodies; (2) while their
liveness analysis (when restricted to synchrony) targets the \emph{sound}
fragment of our analysis, they are more conservative in their approach.
For instance,
the following program which is verified as live in our approach (this
program belongs to $\convm$ in our theory) is judged as a deadlock in
their approach (implemented as {\tt cond-recur} in 
Table~\ref{tbl:results}):
\[
\begin{array}{rl}
\pin{\{X(a,b) = & \pite{\expr_1}{\psend{a}{\expr_2}.X\langle a , b \rangle}{\precv{b}{z}; \zero},\\
Y(a,b) =  & \sel{\precv{a}{z}; Y\langle a , b \rangle,\ \psend{b}{};
  \zero}{}\}}{\\
& \newch{a};\newch{b};
(X\langle a,b\rangle \parr Y\langle a,b\rangle)}
\end{array}
\]
Finally, (3) it is unclear how their tool can deal with the ambiguity of
context sensitive inter-procedural analysis given their use of the
\texttt{oracle} tool and the syntactic approach taken in the implementation.

\paragraph{Behavioural Types} 
Behavioural type-based techniques 
(see~\cite{Huttel:2016} for a broad survey) 
have been developed for general
concurrent program analyses~\cite{DBLP:journals/tcs/IgarashiK04},
such as deadlock-freedom~\cite{Kobayashi06,GiachinoKL14},
lock-freedom~\cite{Kobayashi:2008,Padovani:2014}, resource
usages~\cite{DBLP:journals/lmcs/KobayashiSW06} and 
information flow analysis 
\cite{DBLP:journals/acta/Kobayashi05}. 

All of the type-based techniques above differ from ours in that we
perform an analysis on types akin to \emph{bounded} model-checking,
whereas their works take a type-checking based approach to deadlock
(or lock) freedom.   
Their techniques are sound against all possible inputs of processes, 
but often too conservative.   
Our approach is sound only for some subsets of
possible inputs, but less conservative. 
A potential limitation of their techniques is that subtle
changes in channel usage (that may not have a significant effect on a
program's outcome) can produce significantly different analysis
outcomes (see discussion in~\cite{GiachinoKL14}
and~\cite{Padovani:2014}). Moreover, the
dependency tracking can be quite intricate and hard to implement in a
real language setting.
Our fencing-based approach is more easily implemented as a
\emph{post-hoc} analysis, covering a wide range of Go programs, since
it only limits names in recursive call sites and does not explicitly
depend on the ordering of communications or on computing circularities
of channels (provided the programs are in one of the classes of
\S~\ref{sec:properties}).

The work \cite{GiachinoKL14,KL17} 
develops a deadlock detection analysis
of asynchronous CCS processes with recursion and new name
creation. The analysis is able to reason about infinite-state systems
that create networks with an arbitrary number of processes, going
beyond those of~\cite{Kobayashi06} and~\cite{Padovani:2014}.  Their
approach uses an extension of the typing system of~\cite{Kobayashi06}
as a way to extract a so-called lam term from a (typed) process. Lam
terms track dependencies between channel usages as pairs of level
names. Given a lam term, the authors develop a sound and complete
decision procedure for circularities in dependencies. By separating
this decision procedure from the type system, 
their system is able to accurately analyse deadlock-free processes that are 
not possible in~\cite{Kobayashi06} and~\cite{Padovani:2014}.

We first point out that the deadlock-freedom property of
\cite{GiachinoKL14,KL17} does not match with our notion of liveness
(which is closer to lock freedom
in~\cite{Kobayashi:2008,Padovani:2014}). For instance, their analysis
accepts program $\mathit{Fib}_{\mathit{bad}}$ in
\S~\ref{sec:sync-fib} as a deadlock-free process since a program that
loops non-productively is deadlock-free (but not lock-free).

While their analysis can soundly verify unfenced types, which is by
construction outside the scope of our work, we note that the reduction
of deadlock-freedom to circularity of lams in their work excludes some
natural communication patterns that are finite-control, which can be
soundly checked by our type-level analysis.  Consider the following
finite control program (described in~\cite{KL17}), which can be
directly interpreted as a \fgo type:
\[
\begin{array}{rl}
\pin{\{A(x,y) = & \precv{x}{};\psend{y}{};A\langle x,y\rangle ,\\
B(x,y) =  & \psend{x}{};\precv{y}{};B\langle x , y \rangle \}}{\\
& \newch{a};\newch{b};
  (A\langle a,b\rangle \parr B\langle a,b\rangle)}
\end{array}
\]
The program above consists of two threads that continuously send and
receive along the two channels $a$ and $b$. This program, despite
being finite control (and deadlock-free) is excluded by their
analysis.  As described in~\cite{KL17}, this happens due to their
current formulation of the type system assigning a finite number of
levels in recursive channel usages, which entails that finite-control
systems that use channels infinitely often (such as the one above) can be
assigned circular lams, despite being deadlock-free. Our approach, by
not relying on such notions of circularity can tackle these
finite-control cases in a sound manner.

The work of~\cite{Padovani:2014} studies a variation of 
\cite{Kobayashi:2008,Kobayashi06} 
that ensures deadlock freedom and lock
freedom of the linear $\pi$-calculus,    
with a form of channel polymorphism.  By relying on linearity, the
system in \cite{Padovani:2014} rules out many examples that are
captured by our work (although it can in principle analyse
  unfenced types). The $\mathtt{fib}$ and
$\mathtt{dining philosopher}$
examples denote patterns that are untypable in \cite{Padovani:2014}, but can be
verified in our tool.\fxnote{edited} Our tool can also verify programs
morally equivalent to most examples discussed in~\cite{Padovani:2014},
see Table~\ref{tbl:results} in \S~\ref{sec:eval}.

\paragraph{Session Types} The work on session types is another class of
behavioural typing systems that rely crucially on linearity in channel
types to ensure certain compatibility properties of structured
communication between two (binary
\cite{honda.vasconcelos.kubo:language-primitives}) or more (multiparty
\cite{CHY08}) participants.  Progress (deadlock-freedom on linear
channels) is guaranteed within a \emph{single} session, but not for
multiple interleaved sessions. Several extensions to ensure progress
on multiple sessions have been proposed,
e.g.~\cite{CDYP2014,CDY07,CarboneDM14}.  Our main examples are not
typable in these systems for the same reasons described in the above
paragraph.  Their systems do not ensure progress of shared names,
which are key in our examples.

A different notion of liveness called request-response 
is proposed in \cite{DeboisHSY15} based on binary session types.  Their 
liveness means that when a particular label of a
branching type (a request) is selected, a set of other labels (its
responses) is eventually selected.  The system requires \emph{a priori}
assumptions that a process must satisfy lock-freedom and annotations 
of response labels in types.

The works of \cite{CF10,Wadler2012,DBLP:journals/mscs/CairesPT16}
based on linear logic ensure progress in the presence of multiple
session channels, but the typing discipline disallows modelling of
process networks with cyclic patterns (such as prime sieve).  In these
works, progress denotes both deadlock and lock-freedom in the sense
of~\cite{Padovani:2014}. However, to ensure logical consistency
general recursion is disallowed.  In the presence of general recursion
\cite{DBLP:conf/esop/ToninhoCP13}, progress is weakened; ensuring all
typable processes are deadlock-free but not necessarily lock-free. The
work of \cite{DBLP:conf/tgc/ToninhoCP14} studies a restricted form of
corecursion that ensures both deadlock- and lock-freedom in the
context of logic-based processes. However, since the typing discipline
ensures termination of all computations it is too restrictive for a
more practical setting such as ours.

\paragraph{Effect Systems} 
The work~\cite{Nielson1994} introduces a type and effect system for a
fragment of concurrent ML (including dynamic channel creations
and process spawning) with a predicate on types which guarantees that
typed programs are limited to a finite communication topology.
Their types are only used to check whether a program has a
\emph{finite} communication topology, that is, whether a program uses
a bounded number of channels and a bounded number of processes. No
analysis wrt.\ safe or live communication is given, which is the
ultimate goal of our work.

 \paragraph{Conclusion and Future Work}
\label{sec:conclusion}
Since the early 1990s, 
behavioural type theories which formalise ``types as concurrent processes''
\cite{DBLP:journals/tcs/IgarashiK04} have been studied actively 
in models for concurrency \cite{Huttel:2016}. 
Up to this point, there have been few
opportunities to apply these techniques directly to a real
production-level language. The Go language opens up 
such a possibility. 
This work proposes a static verification framework for liveness and
safety in Go programs based on a bounded execution of \emph{fenced}
behavioural types. 
We develop a tool 
that analyses Go code 
by directly inferring behavioural types
with no need for additional annotations. 

In future work we plan to extend our approach to account for channel
passing, and also lock-based concurrency control, enabling us to
verify \emph{all} forms of concurrency present in Go.  The results of
\S~\ref{sec:types} suggest that it should be possible to encode our
analysis as a model checking problem, allowing us to: (1) exploit the
performance enhancements of state of the art model checking
techniques; (2) study more fine-grained variants of liveness; (3)
integrate model checking into the analysis of conditionals to, in some
scenarios, decide the program class (viz.~\S~\ref{subsec:liveness}).
Another interesting avenue of future work is to explore the
integration of type-checking based approaches into our framework,
including those aimed at liveness and 
termination-checking (such as \cite{DBLP:journals/acta/Kobayashi05, 
Kobayashi:2008,Honda:2007,DBLP:conf/concur/DemangeonHS10}). 
These techniques eliminate false positives arising due to
issues on divergence of processes,     
which are related to our 
classification of \S~\ref{sec:properties}, hence would be useful 
to identify a set of processes which conform, 
e.g.~Proposition~\ref{pro:termination}. 
This would enable a more fine-grained
analysis, taking advantage of the strong soundness properties of this
line of work. Moreover, the latter mentioned works could be applied in
order to soundly approximate the program classification studied in
\S~\ref{subsec:liveness}.

\acks
We gratefuly acknowledge Naoki Kobayashi for finding flaws in
Sections~\ref{sec:types} and~\ref{sec:properties} in earlier versions
of this work; as well as for detailed comments on
Sections~\ref{sec:impl} and~\ref{sec:related}.
The present version aims at correcting these errors. In particular, we
have revised the definition of liveness and safety for types
(Definitions~\ref{def:type-live} and~\ref{def:type-safe}) and removed
theorems related to soundness of our analysis of general liveness and
safety for types. We have revised several remarks in the comparison
between our work
and~\cite{DBLP:journals/acta/Kobayashi05,GiachinoKL14,Kobayashi:2008,Kobayashi06}
(\S~\ref{sec:impl} and~\ref{sec:related}).

We also would like to thank Elena Giachino, Raymond Hu and Luca Padovani 
for fruitful
discussions on this work, as well as the anonymous referees for their
comments.
This work is partially supported by EPSRC EP/K034413/1,
EP/K011715/1, EP/L00058X/1 and EP/N027833/1; and by EU FP7 
612985 (UPSCALE) and COST Action IC1405 (RC).

\balance
\bibliographystyle{abbrv} \bibliography{bib,session.bib}

\newpage 

\appendix
\section{Appendix for Section~\ref{sec:types}}

\subsection{Decidability}

\begin{definition}[Size of $T$]\label{def:size-of-T}
  Define $\tysize{T} \defi  \tysizeG{T}{\emptyset}$, where
  \[
  \tysizeD{T} \defi
  \begin{cases}
    0   & \text{if } T = \zero
    \\
    1 + \tysizeD{T'} & \text{if } T = \alpha ; T',
        \alpha \in \{\typrefix, \End[\tychana] \}
    \\
    1 + \tysizeD{T'} & \text{if } T = (\nu \tychana) T' \text{ or }  {\tynew{\tychana}T'}
    \\
    \sum_{i \in I} \tysizeD{T_i}
    & \text{if } 
    T = {\tch{}{T_i}{i\in I}} \text{ or }
    {\tbr{}{T_i}{i\in I}}
    \\
    \tysizeD{T_1} + \tysizeD{T_2}
    & \text{if } T = T_1 \parr T_2
    \\
    \tysizeG{
      T'\subs{\tilde{\tychana}}{\tilde{\tyvar}} 
    }{\tyunfenv \cup \{\typevara\}}
    & \text{if } 
    T = \tdefcall{\typevara}{\tilde{\tychana}},
    \typevara \notin \tyunfenv,
        \typevara(\tilde{\tyvar}) = {T'}         \\
    0
    & \text{if }   T = \tdefcall{\typevara}{\tilde{\tychana}},  \typevara \in \tyunfenv
  \end{cases}
  \]
\end{definition}

\begin{definition}[Limited Unfolding]\label{def:unfold-T}
  Let $\tyunfenv_k$ be the function from $\domain{\EqTypes}$ to
  $\naturals$ that always returns $k$.
    The $k^{th}$ unfolding of $T$ wrt. $\tilde{\tychana}$, written
  $\tyunfX{k}{\tilde{\tychana}}{T}$
  is given by $\tyunfX{\tyunfenv_k}{\tilde{\tychana}}{T}$, defined below.
    \[
  \tyunf{T} =
  \begin{cases}
    \alpha ; \tyunf{T'} & \text{if } T = \alpha ; T'
    \text{ and }
    \\
    &
    \quad
    \alpha \in \{\tau, \tsend{\tychana}{}, \trecv{\tychana}{}, \End[\tychana] \}
    \\
    0   & \text{if } T = \zero
    \\
    (\nu \tychana)
    \tyunfX
    {\tyunfenv}
    {\tilde{\tychana}}
    {T'} & \text{if } T = (\nu \tychana) T'
    \\
        \tynew{\tychana}
    \tyunfX
    {\tyunfenv}
    {\tilde{\tychana}}
    {T'} & \text{if } T = \tynew{\tychana} T'
    \\
    {\tch{}{\tyunf{T_i}}{i\in I}}
        & \text{if } 
    T = {\tch{}{T_i}{i\in I}}
    \\
    {\tbr{}{\tyunf{T_i}}{i\in I}}
    & \text{if } 
    T = {\tbr{}{T_i}{i\in I}}
    \\
    \tyunf{T_1} \parr \tyunf{T_2} 
    & \text{if } T = T_1 \parr T_2
    \\
    \tyunfX
    {\tyunfenv [ \typevara \mapsto i-1]}
    {\tilde{\tychana}}
    {T' \subs{\tilde{\tychanb}}{\tilde{\tyvar}}}
    & 
    \text{if } 
    T = \tdefcall{\typevara}{\tilde{\tychanb}},
    \tilde{\tychanb} \cap \tilde{\tychana} \neq \emptyset,
        \\
    & \quad  \typevara(\tilde{\tyvar}) = {T'} \in \EqTypes,
        \tyunfenv(\typevara) = i > 0
            \\
    \tdefcall{\typevara}{\tilde{\tychanb}}
    & 
    \text{if } \tyunfenv(\typevara) = 0
    \text{ or }
    \tilde{\tychanb} \cap \tilde{\tychana} = \emptyset
          \end{cases}
  \]
\end{definition}

Lemma~\ref{lem:ty-prec-ind} shows that there cannot be an infinitely
decreasing sequence $\tilde{\tychana}_1 \tytysucc \cdots \tytysucc
\tilde{\tychana}_k \tytysucc \cdots$, when unfolding a type.
We note that due to our convention that in $\typevara(\tilde{x}) = T$,
we have $\tilde{x} \subseteq \fv(T)$, a type $T$ can only remember the
names it received as parameters.
\begin{lemma}\label{lem:ty-prec-ind}
  Given a chain
  \[
  \tilde{\tychana}_1 \tytysucc \tilde{\tychana}_2  \tytysucc \cdots \tytysucc \tilde{\tychana}_k
  \qquad\quad
  k > \lvert \tilde{\tychana}_1 \rvert
  \]
    such that
  $\forall 1 \leq i < k \qst \lvert \tilde{\tychana}_i \rvert = \lvert
  \tilde{\tychana}_{i+1} \rvert$,
    and
    $\forall 1 \leq i < k \qst 
  \tychana \in \tilde{\tychana}_i \land \tychana \notin \tilde{\tychana}_{i+1} 
  \implies
  \forall i < j \leq k \qst  \tychana \notin \tilde{\tychana}_{j}.
  $
      There exists $n \leq \lvert \tilde{\tychana}_1 \rvert$ such that $\forall j \geq n \qst
  \tilde{\tychana}_1 \cap \tilde{\tychana}_{j} = \emptyset$.
\end{lemma}
\begin{proof}
    Relation
  $\tilde{\tychana}
  \tytyprec
  \tilde{\tychanb}$ implies that we must have
  \[
  \tilde{\tychana} = \tychana_{j+1} \cdots \tychana_k \cdot \tychanb_1 \cdots \tychanb_j
  \tytyprec
  \tychana_1 \cdots \tychana_k =  \tilde{\tychanb}
  \]
  with $k \geq 1$ and $j \geq 1$.
        
  Hence, if we remove one element at each step of the chain (the worst
  case), we obtain a chain of the form:
  \[
  \tilde{\tychana}_1 =
  \tychana_1 \cdots \tychana_k 
  \tytysucc
  \tilde{\tychana}_2 =
  \tychana_2 \cdots \tychana_k 
  \cdot
  \tychanb_1
  \tytysucc
  \cdots
  \]
  and, in general: $ \tilde{\tychana}_i = \tychana_i \cdots \tychana_k
  \cdot \tychanb_{1} \cdots \tychanb_{i-1} $ if $i \leq k$; thus we
  have $ \tilde{\tychana}_i \cap \tilde{\tychana_1} = \emptyset$ when
  $ i > k$, as expected, since by assumption we have $\forall 1 \leq i
  < k \qst \tychana \in \tilde{\tychana}_i \land \tychana \notin
  \tilde{\tychana}_{i+1} \implies \forall i < j \leq k \qst \tychana
  \notin \tilde{\tychana}_{j} $ (i.e., once $\tychana$ has been
  removed from $\tilde{\tychana}_i$ it cannot appear again further in
  the chain).
    \end{proof}

Lemma~\ref{lem:ty-bounded-variables} shows that the number of
occurrences of a type variable in the $k$\textsuperscript{th}
unfolding of $T$ is bounded by a function of the size of the syntactical
tree of $T$.

\begin{definition}[Number of Occurrences of $\typevara$ in $T$]\label{def:ocur-in-T}
  \[
  \tyocc{T}{\typevara} =
  \begin{cases}
    0   & \text{if } T = \zero
    \\
    \tyocct{T'} & \text{if } T = \alpha ; T'
    \text{ and }
    \alpha \in \{\typrefix, \End[\tychana] \}
    \\
    \tyocct{T'} & \text{if } T = (\nu \tychana) T'  \text{ or } {\tynew{\tychana}T'}
    \\
    \max(\{\tyocct{T_i}\}_{i \in I})
    & \text{if } 
    T = {\tch{}{T_i}{i\in I}} \text{ or }
    {\tbr{}{T_i}{i\in I}}
    \\
    \tyocct{T_1} + \tyocct{T_2} 
    & \text{if } T = T_1 \parr T_2
    \\
    1
    & \text{if } 
    T = \tdefcall{\typevara}{\tilde{\tychana}} 
   \\
   0
    & \text{if } 
    T = \tdefcall{\typevarb}{\tilde{\tychana}} \text{ and } \typevarb \neq \typevara
      \end{cases}
    \]
\end{definition}

\begin{lemma}\label{lem:ty-bounded-variables}
  If $\tytyJudgeE{\EqTypes}$, then
  for all $\typevara(\tilde{x}) = T \in \EqTypes$,
  \[
  \forall k \geq 0 \qst   
  \tyocc{\tyunfX{k}{\tilde{x}}{T}}{\typevara} 
  \leq
  \tysize{T}^{\lvert \tilde{x} \rvert}
  \]
\end{lemma}
\begin{proof}
  To increase the number of occurrences of a variable, we must have
  parallel processes, which implies that we must strictly decrease the
  number of free names hence eventually reach process that do no use
  any variables in $\tilde{x}$.
  
    It is a simple observation that
  \[
  \forall k \geq 0 \qst   
  \tyocc{\tyunfX{k}{\tilde{x}}{T}}{\typevara} 
  \leq
  \tysize{T}^{k}
  \]
  Hence, the results holds trivially for any $k < {\lvert
    \tilde{x} \rvert}$.

  Assume that the result does not hold by contradiction and take
  $k \geq {\lvert
    \tilde{x} \rvert}$ such that
  \[ 
  \tyocc{\tyunfX{k}{\tilde{x}}{T}}{\typevara} 
  \leq
  \tysize{T}^{\lvert \tilde{x} \rvert}
    \quad \text{and} \quad
    \tyocc{\tyunfX{k+1}{\tilde{x}}{T}}{\typevara} 
  >
  \tysize{T}^{\lvert \tilde{x} \rvert}
  \]

  For the number of occurrences of $\typevara$ to increase strictly,
  we must have 
  \[ 
  T = \tycontxt\left[ 
  \tycontxt_1 [ \tdefcall{\typevara}{\tilde{\tychanb}} ]
  \parr
  \tycontxt_2 [ \tdefcall{\typevara}{\tilde{\tychanc}} ]
  \right]
  \]
  since $\tytyJudgeX{\tilde{x}}{\circ}{\typevara}{T}{\circ}$, we must
  have 
  $ {\tilde{\tychanb}} \tytyprec  {\tilde{x}} $
  and
  $ {\tilde{\tychanc}} \tytyprec  {\tilde{x}} $, 
    which implies, by Lemma~\ref{lem:ty-prec-ind}, that the unfolding
  will terminate.

  In particular, we have a strictly decreasing number of names from
  ${\tilde{x}}$ further down in the tree.
    Thus after $\lvert {\tilde{x}} \rvert$ unfoldings, no names
  from ${\tilde{x}}$ will be left in each occurrences of
  $\tdefcall{\typevara}{\tilde{\tychand}}$ (i.e., ${\tilde{\tychand}}
  \cap {\tilde{x}} = \varnothing$) and $\typevara$ cannot be
  unfolded further.
    \end{proof}

\begin{lemma}\label{lem:ty-freenames-decrease}
  For all types $T$, if $T \equiv T'$ or $T \semty{} T'$, then
  $\freenames{T} \supseteq \freenames{T'}$.
\end{lemma}
\begin{proof}
  Straightforward from the rules in Figure~\ref{fig:types-sem}.
\end{proof}

\lemfinitelts*
\begin{proof}
  Under the symbolic semantics, by
  Lemma~\ref{lem:ty-bounded-variables}, there can be only a limited
  number of occurrences of each variable.
    Hence the processes cannot unfold infinitely wide, therefore, by
  Lemma~\ref{lem:ty-freenames-decrease}, only finitely many names are
  required. If infinitely many names are created, they can be garbage
  collected by structural congruence, in particular
  any type
  $(\nu \tilde{\tychana}) \ \tdefcall{\typevara}{\tilde{\tychanb}}$
  can be replaced by
    $(\nu \tilde{\tychanb}) \ \tdefcall{\typevara}{\tilde{\tychanb}}$
    if $\tilde{\tychana} \supseteq \tilde{\tychanb}$.
      Finally, processes that differ only from their bound names can be
  equated by alpha-equivalence.
  \end{proof}

\thmdecidability*
\begin{proof}
  By Lemma~\ref{lem:finite-lts}, the LTS generated by the $k$-symbolic
  semantics is finite state.
          Hence, checking $k$-liveness and $k$-safety is decidable since the
  $k$-symbolic LTS is finite and thus each term is using only finitely
  many names (after garbage collection).
\end{proof}

\section{Appendix for Section~\ref{sec:properties}}
\label{app:properties}

We use the following lemma for the proofs. 
\begin{lemma}[Inversion]\label{lem:inversion}
  \begin{enumerate}
  \item 
    If $\G \vdash_B P \ts T$ and $P \equiv (\nub c)P'$ then $T \equiv (\nub
    c)T'$, with $\G' \vdash_{B'} P' \ts T'$, for some $\G'$ and $B'$, with
    $\G \subseteq \G'$ and $B \subseteq B'$.
  \item
    If $\G \vdash_s P \ts T$ and $P \equiv P_1\parr P_2$ then $T \equiv
    T_1\parr T_2$, with $\G \vdash_{B_1} P_1 \ts T_1$ and 
    $\G \vdash_{B_2} P_2 \ts T_2$, with $B = B_1 \cup B_2$. 
  \end{enumerate}
\end{lemma}
\begin{proof}
  Straightforward from the typing rules. 
\end{proof}

\thmsubreduction*
\begin{proof}
  We use Lemma \ref{lem:inversion} 
  and similar inversion lemmas for other constructs. 
  Then the rest is straightforward by induction on typing and case analysis on the
  semantics of processes. 
\end{proof}

\thmchansafe*
\begin{proof}
  Suppose $X_0\ENCan{}\tra{}^*(\nub \tilde{c}){Q}$ and 
  $Q\barb{\cclbarb{a}}$. 
  Then by Lemma~\ref{lem:inversion},  
  there exists $T$ such that 
  $\G' \vdash Q \ts T$. By Lemma~\ref{lem:barb}, 
  we have $T\barb{\cclbarb{a}}$. Since $T$ is safe, 
  $\neg(T\wbarbt{\clbarb{a}})$ and $\neg (T\wbarbt{\ov{a}})$.
  This implies, by applying Lemma~\ref{lem:barb} again,  
  $\neg(Q\wbarb{\clbarb{a}})$ and $\neg (Q\wbarb{\ov{a}})$.
\end{proof}

\proptermination*
\begin{proof}
  Assume $\G \vdash \prog \ts \EqTypes$ and 
  $\EqTypes$ is live. \\[1mm]
  {\bf (1)} Suppose by contradiction that 
  $X_0\ENCan{}\tra{}^\ast P \not\tra{}$ but $P\not\equiv \zero$. 
  Then there exists $Q$ such that $P\equiv (\nu\tilde{a})(Q)$ and 
  $Q\barb{b}$ or $Q\barb{\tilde{b}}$ for some $b$ or $\tilde{b}$. 
  Then by Lemma \ref{lem:barb}, it contradicts $\EqTypes$ is live. \\[1mm]
  {\bf (2)} By definition of liveness (there is always a path
  ($\tau$-actions) to reach the term $\zero$, which is live). 
\end{proof}

\propliveness*
\begin{proof}
  Suppose $\prog\not\in \ic$
  and $X_0\ENCan{} \tra{}^\ast (\nu \vec{a})P$. 
  Then by Inversion Lemma (Lemma~\ref{lem:inversion}), there exists $\G' \vdash_B P \ts T$. 
  Then $P \wbarb{b}$ iff $T\wbarbt{b}$ and $P \wbarb{\tilde{b}}$ iff
  $T\wbarbt{\tilde{b}}$ since  the reduction of $T$ coincides with 
  the reduction of $P$.  Hence $\prog$ is live. 
\end{proof}

\paragraph{Proof of Theorem~\ref{thm:liveness}}

\begin{lemma}\label{lem:noif_live}
Suppose $\G \vdash \prog \ts \EqTypes$,  $\EqTypes$ is live and
$\prog$ is conditional-free, then $\prog$ is live.
\end{lemma}
\begin{proof}
Straightforward from the fact that all process moves are
\emph{strongly} matched by type moves and $\EqTypes$ is live.
\end{proof}

\begin{lemma}\label{lem:allif_live}
Suppose $\G \vdash \prog \ts \EqTypes$, $\EqTypes$ is live, $\prog \in
\ac$ and $\prog$ and all marked conditionals in $\m{mark}(\prog)$ are in
$\infcond$, then $\prog$ is live.
\end{lemma}
\begin{proof}
Since all conditionals in $\prog$ are infinite,
then by the definition of $\ac$ and $\EqTypes$ being live gives us
that $P\wbarb{o}$ iff $T\wbarbt{o}$, from which the result follows.
\end{proof}

\begin{lemma}\label{lem:finifaway}
Suppose $\G \vdash \prog \ts \EqTypes$, $\prog \in \ac$, $\EqTypes$ is live and
$\m{mark}(\prog)$ has a marked branch $n$ that is not in
$\infcond(\prog)$. Then $\prog$ has an infinite trace in which $n$ no
longer appears.
\end{lemma}
\begin{proof}
By $\prog \in \ac$ it follows that $\prog$ cannot have a finite trace,
thus the result follows by $n \not\in \infcond(\prog)$.
\end{proof}

\thmliveness*

\begin{proof}
If all conditionals in $\prog$ are infinite, then the result follows
by Lemma~\ref{lem:allif_live}. Otherwise, by Lemma~\ref{lem:finifaway}
we can execute $\prog$ until there are no more finite conditionals,
from which then the result follows by Subject Reduction
(Theorem~\ref{thrm:subred}) and Lemma~\ref{lem:allif_live}.
\end{proof}

\section{Appendix for Section~\ref{sec:asynchrony} (Asynchrony)}
\label{app:async}

\begin{definition}[Asynchronous type barbs]~\label{def:async-type-barbs}
\[
\inferrule
{k < n}
{{\astyOpenbuf{\tychana}{k}{n}} \barb{\astyBufSend{\tychana}}}
\quad
\inferrule
{k \leq 1}
{{\astyOpenbuf{\tychana}{k}{n}} \barb{\astyBufRcv{\tychana}}}
\quad
\inferrule
{
  T\barb{\ov{\tychana}}
  \quad
  T'\barb{\astyBufSend{\tychana}}}
{T\parr T'\barb{\labsync{\tychana}}}
\quad
\inferrule
{
  T\barb{{\astyBufRcv{\tychana}}}
      \quad \typrefix_i \barb{\tychana}
}
{T\parr \tbr{}{\typrefix_i ; S_i}{i \in I} \barb{\labsync{\tychana}}}
\]
\end{definition}

\paragraph{Proofs of Theorems \ref{thm:decidability-aty} and \ref{thm:a:safety:liveness} }
Since the definitions of liveness and safety stay the same, and since
the semantics are essentially isomorphic (the symbolic semantics is
unaffected, notably) the proofs are essentially identical with the
synchronous cases.

 \section{Appendix for Section~\ref{sec:impl}}

Below we present source code of the examples we created for our evaluation.

\subsection{forselect}
The \texttt{forselect} program spawns two communicating for-select loop and
both goroutines terminate when channel \lstgo{term} is used.
The choice of sending or receiving from the channels inside the
for-select loop is non-deterministically decided by the \lstgo{select} primitive
in Go.
\begin{lstlisting}[language=Go,style=golang]
package main

import "fmt"

func sel1(term, ch chan int, done chan struct{}) {
	for {
		select {
		case <-term: // Receive termating message.
			fmt.Println("sel1: recv")
			done <- struct{}{}
			return
		case ch <- 1:
			fmt.Println("sel1: send")
		}
	}
}

func sel2(term, ch chan int, done chan struct{}) {
	for {
		select {
		case <-ch:
			fmt.Println("sel2: recv")
		case term <- 2: // Send terminating message.
			fmt.Println("sel2: send")
			done <- struct{}{}
			return
		}
	}
}

func main() {
	done := make(chan struct{})
	term := make(chan int) // Terminating channel.
	data := make(chan int)
	go sel1(term, data, done)
	go sel2(term, data, done)

	<-done
	<-done
}
\end{lstlisting}

\subsection{cond-recur}
The \texttt{cond-recur} program is similar to \texttt{forselect} but the
decision of when to terminate the for-select loop is determined by a
conditional in the \lstgo{x} goroutine, rather than a terminating channel.
\begin{lstlisting}[language=Go,style=golang]
package main

import "fmt"

func x(ch chan int, done chan struct{}) {
	i := 0
	for {
		if i < 3 { // This condition decides when to terminate.
			ch <- i
			fmt.Println("Sent", i)
			i++
		} else {
			done <- struct{}{} // Send terminate message.
			return             // Break out of loop.
		}
	}
}

func main() {
	done := make(chan struct{})
	ch := make(chan int)
	go x(ch, done) // Spawn a decision making goroutine.
FINISH:
	for {
		select {
		case x := <-ch:
			fmt.Println(x)
		case <-done:   // Terminate message.
			break FINISH // Break out of loop.
		}
	}
}
\end{lstlisting}

\subsection{jobsched}
The \texttt{jobsched} program sets up a shared job queue between two workers to
process the incoming jobs.
\begin{lstlisting}[language=Go,style=golang]
package main

import (
	"fmt"
	"time"
)

var i int

func worker(id int, jobQueue <-chan int, done <-chan struct{}) {
	for {
		select {
		case jobID := <-jobQueue:
			fmt.Println(id, "Executing job", jobID)
		case <-done:
			fmt.Println(id, "Quits")
			return
		}
	}
}

func morejob() bool {
	i++
	return i < 20
}

func producer(q chan int, done chan struct{}) {
	for morejob() {
		q <- i
	}
	close(done)
}

func main() {
	jobQueue := make(chan int)
	done := make(chan struct{})
	go worker(1, jobQueue, done)
	go worker(2, jobQueue, done)
	producer(jobQueue, done)
	time.Sleep(1 * time.Second)
}
\end{lstlisting}
 
\section{Additional Discussion of Related Work}
\label{app:related}
We discuss some additional related work. 

{\bf Session types.}
Another key limitation of session typing
is that conditional branching is usually typed with the following
rule:\\[1mm]
\hspace*{1.5cm}
$
\inferrule[]
{\G \vdash e {:}\m{bool} \quad \G \vdash P \ts \D \quad \G \vdash Q \ts \D}
{\G \vdash \pite{e}{P}{Q} \ts \D}
$\\[1mm]
The rule above (even in the presence of subtyping, e.g.~\cite{GH05})
enforces that the communication behaviour of the two branches must be
fundamentally the same. This turns out to be too restrictive in
practice, where branching is mostly used to define the conditions
under which behaviour should indeed be morally different.

{\bf Model Checking.} 
The classical work~\cite{DBLP:conf/popl/ChakiRR02} verifies progress properties
($\Box \varphi$) of the $\pi$-calculus 
applying a LTL model checking tool to types which take the form 
of CCS terms. To limit state-space explosion, the work relies 
on an assume-guarantee reasoning technique. The work 
requires ``user input'' type signatures (annotations) on 
processes. Our framework does not require such annotations.

The work \cite{DOsualdoKO13} studies a verification of Erlang programs
where processes communicate via unbounded mailboxes 
and can be spawned dynamically and potentially infinitely.
Erlang programs are modelled as vector addition systems (VAS)
and a VAS-based tool is used to check 
simple reachability properties. The model is limited to a bounded number of channels, 
hence the topologies represented in their work are more limited 
than those induced by fenced types.

The work~\cite{Yasukata0M14} uses higher-order model checking to
verify a concurrent calculus which features dynamic process creation.
They transform processes into higher-order recursion schemes, 
which in turn generates action trees on which 
they can check whether two threads enter the same critical section. 
Because the action trees model executions only at the single process level
(i.e.~ignoring how processes are interleaved and communicate with
each other), it is not straightforward to adapt their approach 
to verify our liveness and safety properties.

\end{document}